\documentclass[11pt]{article}

\usepackage[letterpaper,margin=0.9in]{geometry}
\usepackage{amsmath, amssymb, amsthm, amsfonts}

\usepackage{ifthen}
\usepackage{arrayjobx}
\usepackage{tikz}
\usetikzlibrary{positioning,decorations.pathreplacing}

\usepackage{cite}
\usepackage{appendix}
\usepackage{graphicx}
\usepackage{color}
\usepackage{algorithm}
\usepackage[noend]{algpseudocode}
\usepackage{epstopdf}
\usepackage{float}
\usepackage{pgf}
\usepackage[framemethod=tikz]{mdframed}
\usepackage[bottom]{footmisc}
\usepackage{enumitem}
\setitemize{noitemsep,topsep=0pt,parsep=0pt,partopsep=0pt}
\usepackage{caption}
\usepackage{xspace}
\usepackage{hyperref}
\usepackage[capitalize]{cleveref}
\crefformat{footnote}{#2\footnotemark[#1]#3}
\algnewcommand\algorithmicswitch{\textbf{switch}}
\algnewcommand\algorithmiccase{\textbf{case}}

\algdef{SE}[SWITCH]{Switch}{EndSwitch}[1]{\algorithmicswitch\ #1\ \algorithmicdo}{\algorithmicend\ \algorithmicswitch}%
\algdef{SE}[CASE]{Case}{EndCase}[1]{\algorithmiccase\ #1}{\algorithmicend\ \algorithmiccase}%
\algtext*{EndSwitch}%
\algtext*{EndCase}%

\newcommand{\eps}{\varepsilon}
\newcommand{\LBM}{$\mathsf{V}$-$\mathsf{CONGEST}$\xspace}
\newcommand{\CM}{$\mathsf{E}$-$\mathsf{CONGEST}$\xspace}

\newcommand{\ceil}[1]{\lceil #1 \rceil}

\newcommand{\NumOfLayers}{L}
\newcommand{\set}[1]{\left\{#1\right\}}

\DeclareMathOperator{\E}{\mathbb{E}}

\newcommand{\componentID}{\mathit{componentID}}
\newcommand{\nlist}{\mathit{List}}
\newcommand{\acceptedProposal}{\mathit{acceptedProposal}}

\renewcommand{\paragraph}[1]{\vspace{0.15cm}\noindent {\bf #1}:}
\newtheorem{theorem}{Theorem}[section]
\newtheorem{lemma}[theorem]{Lemma}

\newtheorem{remark}[theorem]{Remark}
\newtheorem{corollary}[theorem]{Corollary}
\newtheorem{proposition}[theorem]{Proposition}

\newcommand{\FullOrShort}{short}

\ifthenelse{\equal{\FullOrShort}{full}}{
	
	  \newcommand{\fullOnly}[1]{#1}
	  \newcommand{\shortOnly}[1]{}
		\newcommand{\IncludePictures}[1]{#1}

  }{
    \renewcommand{\baselinestretch}{0.99}
	  \newcommand{\fullOnly}[1]{}
	  \newcommand{\shortOnly}[1]{#1}
		\newcommand{\IncludePictures}[1]{}

  }

\begin{document}

\date{}

\title{Distributed Connectivity Decomposition}

\author{
  Keren Censor-Hillel\\
  \small Technion
  \thanks{\small ckeren@cs.technion.ac.il, Shalon fellow.}
  \and
  Mohsen Ghaffari\\
  \small MIT
   \thanks{\small ghaffari@csail.mit.edu}
  \and
  Fabian Kuhn\\
  \small University of Freiburg
  \thanks{\small kuhn@cs.uni-freiburg.de}
}

\maketitle
\vspace{-.7cm}
\setcounter{page}{0}
\thispagestyle{empty}

\begin{abstract}
A fundamental problem in distributed network algorithms is to obtain information flow matching the connectivity. Despite ingenious ideas such as \emph{network coding}, this goal remains rather elusive.

In this paper, we present time-efficient distributed algorithms for decomposing graphs with large edge or vertex connectivity into multiple spanning or dominating trees, respectively. These decompositions allow us to achieve a flow with size close to the connectivity by parallelizing it along the trees. More specifically, our distributed decomposition algorithms are as follows:
\begin{itemize}[noitemsep,topsep=5pt,parsep=5pt,partopsep=0pt]
\item[(I)] A decomposition of each undirected graph with \emph{vertex-connectivity} $k$ into (fractionally) \emph{vertex-disjoint} weighted \emph{dominating trees} with total weight $\Omega(\frac{k}{\log n})$, in $\widetilde{O}(D+\sqrt{n})$ rounds.
\item[(II)] A decomposition of each undirected graph with \emph{edge-connectivity} $\lambda$ into (fractionally) \emph{edge-disjoint} weighted \emph{spanning trees} with total weight $\ceil{\frac{\lambda-1}{2}}(1-\eps)$, in $\widetilde{O}(D+\sqrt{n\lambda})$ rounds.
\end{itemize}
We also show round complexity lower bounds of  $\tilde{\Omega}(D+\sqrt{\frac{n}{k}})$ and $\tilde{\Omega}(D+\sqrt{\frac{n}{\lambda}})$ for the above two decompositions, using techniques of [Das Sarma et al., STOC'11]. Our vertex-connectivity decomposition also extends to centralized algorithms and improves the time complexity of [Censor-Hillel et al., SODA'14] from $O(n^3)$ to near-optimal $\tilde{O}(m)$.

\bigskip
\noindent As corollaries of our decompositions, we also get the following results:
\begin{enumerate}[noitemsep,topsep=5pt,parsep=5pt,partopsep=0pt]
\item Distributed \emph{oblivious routing broadcast} with $O(1)$-competitive edge-congestion and $O(\log n)$-competitive vertex-congestion. We find the latter more interesting as, although centralized oblivious routings with $O(\log n)$-competitive edge-congestion are known [R\"{a}cke, STOC'08], no \emph{point-to-point} oblivious routing can have $o(\sqrt{n})$ vertex-congestion competitiveness.

\item A centralized $\tilde{O}(m)$ time $O(\log n)$-approximation for vertex connectivity. This moves towards the 1974 conjecture of Aho, Hopcroft and Ullman that postulates the existence of an $O(m)$ exact algorithm. Currently the state of the art is an $O(n^2k+ \min\{nk^{3.5}, n^{1.75} k^2\})$ exact algorithm~[Gabow, FOCS'00] and an $O(\min\{n^{2.5}, n^2k \})$ $2$-approximation~[Henzinger, JALG'97].

\item A distributed $O(\log n)$-approximation of vertex connectivity with round complexity of $\widetilde{O}(D+\sqrt{n})$. Designing distributed algorithms for approximating or computing vertex connectivity had remained open and this is the first (non-trivial) answer.
\end{enumerate}

\end{abstract}
\newpage

\section{Introduction and Related Work}\label{sec:Intro}
Edge and vertex connectivity are two basic graph-theoretic measures. One important application of these measures is in transferring information between nodes of a network, which is the ultimate goal of communication networks and also a central issue in distributed computing\cite[Section 1.3.1]{peleg2000distributed}. Edge and vertex connectivity characterize limits on the global information flow as each edge or vertex cut defines an upper bound on the flow across the cut. Hence, we naturally expect networks with larger connectivity to provide a better communication medium and support larger information flow. However, designing distributed algorithms that leverage large connectivity has remained elusive.


The importance of achieving better information flows is exemplified by \emph{network coding}, which has received extensive attention both in theory and in practice (see~\cite{ACLY} and citations thereof). 
One of the main attractions of network coding is that, if we ignore the overhead due to coding coefficients, it usually achieves an information flow almost matching the size of the minimum cut~\cite{ACLY}. However, in standard distributed networks each message can contain at most $O(\log n)$ bits\cite{peleg2000distributed} and thus, because of the coefficients, network coding can only support a flow of $O(\log n)$ messages per round. 

This paper provides a distributed solution for exploiting large connectivity for the goal of obtaining a large flow of information.
We take a rather natural approach, which we call \emph{connectivity decomposition}, of decomposing connectivity into smaller and more manageable units. We present distributed algorithms that decompose a graph with a large connectivity into many spanning or dominating trees\footnote{\label{note1}A tree $H=(V_{T}, E_{T})$ is a \emph{dominating tree} of $G=(V_{G}, E_{G})$ if $V_{T}\subseteq V_{G}$, $E_{T}\subseteq E_{G}$, and each node in $V_{G}\setminus V_{T}$ has a $G$-neighbor in $V_{T}$. Note that when we want many vertex-disjoint subgraph, unavoidably we have to relax the \emph{``spanning"} property and the \emph{``dominating"} condition is the natural alternative.}, while almost preserving the total connectivity, providing a platform for utilizing large connectivity. For instance, for the goal of information dissemination we can now parallelize the information flow along the trees and get a flow size close to the connectivity. 

An interesting comparison of our approach of handling \emph{congestion} is to that of addressing \emph{locality}, as these are two core issues in distributed network algorithms\cite[Section 2.3]{peleg2000distributed}. While locality has been studied extensively and many general techniques are known for dealing with it, the methods used for handling congestion in different problems appear to be more ad hoc. A fundamental and generic technique centered around locality is \emph{locality-based decompositions}~\cite{peleg2000distributed}, which group nodes in small-diameter clusters with certain properties; classical examples include \cite{Awerbuch89, Awerbuch:1992, awerbuch1996fast, Panconesi:1992, KuttenPeleg95}. In this regard, we can view connectivity decompositions as a direction essentially orthogonal to that of the \emph{locality-based} decompositions and they form a systematic step towards addressing congestion.

The rest of this section is organized as follows: We first briefly explain the relation between tree packings and connectivity in \Cref{subsec:recap}. Then, we present our decomposition results and their applications in Sections \ref{subsec:results} and \ref{subsec:applications}, respectively. We review some other related work and specially the centralized connectivity decompositions results in \Cref{subsec:related}.

\subsection{Connectivity and Tree Packings}\label{subsec:recap}
Menger's theorem (see~\cite[Chapter 9]{bondymurty})---which is the most basic result about graph connectivity---tells us that, in each graph with edge connectivity $\lambda$ or vertex connectivity $k$, each pair of vertices are connected via $\lambda$ edge-disjoint paths or $k$ internally vertex-disjoint paths, respectively. However, when we have to deal with more than two nodes, this theorem is insufficient especially because it does not inform us about the structure of overlaps between paths of different vertex pairs. 

Spanning and dominating tree packings allow us to manage this overlap and provide a medium for decomposing edge and vertex connectivity, respectively, into their single units: If we find a collection of $\lambda'$ edge-disjoint spanning trees---which we call a \emph{spanning tree packing of size $\lambda'$}---then for each pair of vertices we get $\lambda'$ edge-disjoint paths, one through each tree. More importantly, for any number of vertex pairs, the paths going through different trees are edge-disjoint. Similarly, if we have $k'$ vertex-disjoint dominating trees\cref{note1}---which we call a \emph{dominating tree packing} of size $k'$--- then for each vertex pair we get $k'$ internally vertex-disjoint paths, one through each tree. More importantly, for any number of pairs, the paths going through different trees are internally vertex-disjoint. 

In both spanning and dominating tree packings, we can relax the disjointness requirement to \emph{fractional disjointness}. That is, we allow the trees to overlap but now each tree $\tau$ has a weight $\omega_{\tau} \in [0,1]$ and the total weight in each edge or vertex, respectively, has to be at most $1$. This naturally corresponds to sharing the edge or vertex between the trees proportional to their weights (e.g., time-sharing in information dissemination). 

Edge connectivity decompositions have been known for a long time, thanks to beautiful and highly non-trivial (existential) results of Tutte~\cite{Tutte} and Nash-Williams \cite{Nash-Williams} from 1960. These results show that each graph with edge-connectivity $\lambda$ contains a spanning tree packing of $\ceil{\frac{\lambda-1}{2}}$ (see~\cite{Kundu}). This bound is existentially tight even for the fractional version and has numerous important applications. 

Vertex connectivity decompositions were addressed only recently: \cite{CGK14} showed that each graph with vertex-connectivity $k$ contains a dominating tree packing of $\Omega({\kappa/\log^2 n})$ and a fractional dominating tree packing of $\Omega({k/\log n})$. Here, $\kappa$ is the vertex-connectivity of the sampled graph when each vertex is sampled with probability $1/2$, and the proven bound of $\kappa=\Omega({k/\log^3 n})$ is currently the best known. The paper also showed that the $\Omega({k/\log n})$ fractional packing bound is existentially tight.

\subsection{Our Results}\label{subsec:results}
We present distributed algorithms that provide edge and vertex connectivity decompositions, which are comparable to their centralized counterparts, while having efficient (or near-optimal) round complexity.

We note that, although our edge-connectivity decomposition builds on a number of standard techniques and known results, our vertex connectivity decomposition algorithm is the main technical novelty of this paper and interestingly, it achieves near-optimal time complexities in both distributed and centralized settings (\Cref{distDomPack} and \Cref{centDomPack}). 

\smallskip
For distributed settings, we consider two synchronous message passing models: \LBM, where in each round, each node can send one $O(\log n)$-bits message to \emph{all} of its
neighbors, and \CM, where in each round, one $O(\log n)$-bits message can be sent in each direction of each edge. As the names suggest, the congestion in the \LBM model is in the vertices whereas in \CM, it is in the edges. We note that \CM, often called $\mathcal{CONGEST}$, is the classical distributed model that considers \emph{congestion} and has bounded size messages~\cite{peleg2000distributed}. Furthermore, \LBM is a restricted version of \CM and thus any algorithm working in \LBM works in \CM as well.

In all our results, we assume a connected undirected network with $n$ nodes, $m$ edges, and diameter $D$. Moreover, we usually use $k$ for vertex connectivity and $\lambda$ for edge connectivity.

\begin{theorem}\label{distDomPack} There is an $\tilde{O}(\min\{D+\sqrt{n}, \frac{n}{k}\})$-rounds randomized distributed algorithm in the \LBM model that w.h.p.\footnote{We use the phrase with high probability (w.h.p.) to indicate a probability at least $1-\frac{1}{n^c}$ for some constant $c\geq 1$.} finds a fractional dominating tree packing of size $\Omega(\frac{k}{\log n})$, where $k$ is the vertex connectivity of the graph. More specifically, this algorithm finds $\Omega(k)$ dominating trees, each of diameter $\tilde{O}(\frac{n}{k})$, such that each node is included in $O(\log n)$ trees. 
\end{theorem}

\begin{theorem}\label{centDomPack} There is an $\tilde{O}(m)$ time randomized centralized algorithm that w.h.p. finds a fractional dominating tree packing of size $\Omega(\frac{k}{\log n})$, where $k$ is the vertex connectivity of the graph. More specifically, this algorithm finds $\Omega(k)$ dominating trees, each of diameter $\tilde{O}(\frac{n}{k})$, such that each node is included in $O(\log n)$ trees. 
\end{theorem}


\Cref{centDomPack} improves over the (at least) $\Omega(n^3)$ algorithms of \cite{CGK14} and \cite{Alina}. Regarding \Cref{distDomPack}, we note that the algorithm of~\cite{Alina} does not appear to admit a distributed implementation as it is based on a number of centralized tools and techniques such as the ellipsoid method of linear programs, the meta-rounding of~\cite{CarrVempala}, and the Min-Cost-CDS approximation result of~\cite{GuhaKhuller}. The algorithm of \cite{CGK14} has a similar problem which indeed seems essential and unavoidable. See the last part of \Cref{subsec:outline} for an intuitive explanation of why the algorithm of \cite{CGK14} does not extend to a distributed setting and for how it compares with the approach in this paper. 



\begin{theorem}\label{distSpanPack} There is an $\tilde{O}(D+\sqrt{\lambda n})$-rounds randomized distributed algorithm in the \CM model that w.h.p. finds a fractional spanning tree packing of size $\ceil{\frac{\lambda-1}{2}} (1-\eps)$, where $\lambda$ is the edge connectivity of the graph. Furthermore, each edge is included in at most $O(\log^3 n)$ trees.
\end{theorem}

  
\medskip
\paragraph{Integral Tree Packings} Both algorithms of \Cref{distDomPack} and \Cref{centDomPack} can be adapted to produce a dominating tree packing of size $\Omega(\frac{\kappa}{\log^2 n})$, in similar time-complexities, using the random layering technique of the (proof of) \cite[Theorem 1.2]{CGK14}. Here, $\kappa$ is the remaining vertex-connectivity when each vertex is sampled with probability $1/2$, for which currently the best known bound is $\kappa=\Omega(\frac{k}{\log^3 n})$~\cite{CGK14}. Also, a considerably simpler variant of the algorithm of \Cref{distSpanPack} can be adapted to produce a spanning tree packing of size $\Omega(\frac{\lambda}{\log n})$, with a similar $\widetilde{O}(D+\sqrt{\lambda n})$ round complexity.

We note that for information dissemination purposes---which is the primary application of our decompositions---fractional packing is as useful as integral packing. This is because we can easily timeshare each node or edge amongst the trees that use it proportional to their weights. Since our fractional packing results have better sizes, our main focus will be on describing the fractional versions.  

\medskip
\paragraph{Lower Bounds}
While $\Omega(m)$ is a trivial lower bound on time-complexity of the centralized decompositions, by extending results of \cite{dassarma12}, we show lower bounds of $\tilde{\Omega}(D+\sqrt{\frac{n}{k}})$  and $\tilde{\Omega}(D+\sqrt{\frac{n}{\lambda}})$ on the round complexities of the above distributed decompositions, i.e., respectively \Cref{distDomPack} and \Cref{distSpanPack}. See \Cref{sec:lower} for the formal statements and the proofs. 

Furthermore, as an interesting comparison, we note that in the model \LBM (or \CM), in a network with vertex connectivity $k$ (resp., edge-connectivity $\lambda)$, simply learning the ids of the nodes in the $2$-neighborhood might require $\Omega(\frac{n}{k})$ rounds (resp., $\Omega(\frac{n}{\lambda})$ rounds), if a node has $k$ neighbors and $n-k-1$ nodes at distance $2$ (resp., $\lambda$ neighbors and $n-\lambda-1$ nodes at distance $2$).

\subsection{Applications}\label{subsec:applications}
\subsubsection{Applications to Information Dissemination}
As their primary application, our decompositions provide time-efficient distributed constructions for routing-based broadcast algorithms with existentially optimal throughput. See \Cref{app:example} for a simple example. Note that in the \LBM model, vertex cuts characterize the main limits on the information flow, and $k$ message per round is the clear information-theoretic limit on the broadcast throughput (even with network coding) in each graph with vertex connectivity $k$. Similarly, in the \CM model, edge cuts characterize the main limits on the information flow, and $\lambda$ message per round is the information-theoretic limit on the broadcast throughput (even with network coding) in each graph with edge connectivity $\lambda$. Our optimal-throughput broadcast algorithms are as follows:

\begin{corollary} In the \LBM model, using the $\tilde{O}(D+\sqrt{n})$-rounds construction of \Cref{distDomPack}, and then broadcasting each message along a random tree, we get a broadcast algorithm with throughput of $\Omega(\frac{k}{\log n})$ messages per round. \cite{CGK14} shows this throughput to be existentially optimal. 
\end{corollary}

See \Cref{app:example} for an explanation and a simple example of how one can use dominating tree packings to broadcast messages by routing them along different trees. 

\begin{corollary} In the \CM model, using the $\tilde{O}(D+\sqrt{\lambda n})$-rounds construction of \Cref{distSpanPack}, and then broadcasting each message along a random tree, we get a broadcast algorithm with throughput of $\ceil{\frac{\lambda-1}{2}} (1-\eps)$ messages per round. 
\end{corollary}

We emphasize that the above broadcast algorithms provide \emph{oblivious routing} (see~\cite{Racke:2002}). In an \emph{oblivious routing} algorithm, the path taken by each message is determined (deterministically or probabilistically) independent of the current load on the graph; that is, particularly independent of how many other messages exist in the graph and how they are routed. Note that this is in stark contrast to (offline) adaptive algorithms which can tailor the route of each message, while knowing the current (or future) load on the graph, in order to minimize congestion. Quite surprisingly, as a tour de force of a beautiful line of work~\cite{Racke:2002, Azar:2003, Harrelson:2003, Bienkowski:2003, Racke:2008}, R\"{a}cke~\cite{Racke:2002} presented a centralized oblivious routing algorithm with $O(\log n)$-competitive edge-congestion. That is, in this algorithm, the expected maximum congestion over all edges is at most $O(\log n)$ times the maximum congestion of the offline optimal algorithm. The problem of designing distributed oblivious routing algorithms achieving this performance remains open. Furthermore, it is known that no \emph{point-to-point} oblivious routing can have vertex-congestion competitiveness better than $\Theta(\sqrt{n})$ \cite{Hajiaghayi07}. Our results address \emph{oblivious routing for broadcast}:

\begin{corollary}
By routing each message along a random one of the trees generated by \Cref{distDomPack} and \Cref{distSpanPack}, we get distributed oblivious routing broadcast algorithms that respectively have vertex-congestion competitiveness of $O(\log n)$ and edge-congestion competitiveness of $O(1)$. 
\end{corollary}

\subsubsection{Applications on Vertex Connectivity Approximation}\label{subsubsec:VCapprox}
Vertex connectivity is a central concept in graph theory and extensive attention has been paid to developing algorithms that compute or approximate it. In 1974, Aho, Hopcraft and Ulman~\cite[Problem 5.30]{aho1974design} conjectured that there should be an $O(m)$ time algorithm for computing the vertex connectivity. Despite many interesting works in this direction---e.g., ~\cite{Tarjan74, Even75, Galil80, Henzinger97, Henzinger1996, Gabow00}---finding $O(m)$ time algorithms for vertex connectivity has yet to succeed. The current state of the art is an $O(\min\{n^{2} k+nk^{3.5}, n^2 k + n^{1.75} k^2\})$ time exact algorithm by Gabow~\cite{Gabow00} and an $O(\min\{n^{2.5}, n^2k \})$ time $2$-approximation by Henzinger~\cite{Henzinger97}. The situation is considerably worse in distributed settings and the problem of upper bounds has remained widely open, while we show in \Cref{sec:lower} that an $\Omega(D+\sqrt{\frac{n}{k}})$ round complexity lower bound follows from techniques of \cite{dassarma12}. 

Since  \Cref{distDomPack} and \Cref{centDomPack} work without a priori knowledge of vertex connectivity and as the size of the achieved dominating trees packing is in the range $[\Omega(k/\log n), k]$, our dominating tree packing algorithm provides the following implication:
 
\begin{corollary}
We can compute an $O(\log n)$ approximation of vertex connectivity, via a centralized algorithm in $\tilde{O}(m)$ time, or via a distributed algorithm in $\tilde{O}(D+\sqrt{n})$ rounds of the \LBM model.
\end{corollary}
Note that it is widely known that in undirected graphs, vertex connectivity and vertex cuts are significantly more complex than edge connectivity and edge cuts, for which now the following result are known: an $\tilde{O}(m)$ time centralized exact algorithm\cite{KargerSTOC96}, an $O(m)$ time centralized $(1+\eps)$-approximation\cite{KargerSODA94}, and an $\tilde{O}(D+\sqrt{n})$ rounds distributed $(2+\eps)$-approximation\cite{GK13}.

\subsection{Other Related Work}\label{subsec:related}
\subsubsection{Independent Trees}\label{subsubsec:IndependentTrees}
Dominating tree packings have some resemblance to \emph{vertex independent trees}\cite{zehavi1989three, khuller1992independent} and are in fact a strictly stronger concept. In a graph $G=(V, E)$, $k'$ trees are called \emph{vertex independent} if they are spanning trees all rooted in a node $r\in V$ and for each vertex $v\in V$, the paths between $r$ and $v$ in different trees are internally vertex-disjoint. We emphasize that these trees are not vertex disjoint.

Zehavi and Itai conjectured in 1989~\cite{zehavi1989three} that each graph with vertex connectivity $k$ contains $k$ vertex independent trees. Finding such trees, if they exist, is also of interest. The conjecture remains open and is confirmed only for cases $k\in \{2,3\}$. Itai and Rodeh~\cite{itai1988multi} present an $O(m)$ time centralized algorithm for finding $2$ vertex independent trees, when the graph is $2$-vertex-connected and Cheriyan and Maheshvari~\cite{cheriyan1988finding} present an $O(n^2)$ time centralized algorithm for finding $3$ vertex independent trees, when the graph is $3$-vertex-connected. 

Vertex disjoint dominating trees are a strictly stronger notion\footnote{The relation is strict as e.g., the following graph with vertex connectivity $3$ does not admit more than $1$ vertex-disjoint dominating trees while~\cite{cheriyan1988finding} implies that this graph has $3$ vertex independent trees: A graph with a clique of size $n^{1/3}$, plus one additional node for each subset of three nodes of the clique, connected exactly to those three clique nodes.}: Given $k'$ vertex-disjoint dominating trees, we get $k'$ vertex independent trees, for \emph{any} root $r\in V$. This is by adding all the other nodes to each dominating tree as leaves to make it spanning. Then, for each vertex $v\in V$, the path from $r$ to $v$ in each (now spanning) tree uses only internal vertices from the related dominating tree. 

In this regard, one can view \cite[Theorem 1.2]{CGK14} as providing a poly-logarithmic approximation of the Zehavi and Itai's conjecture. Furthermore, the vertex connectivity algorithm presented here (formally, its extension to integral dominating tree packing, mentioned in \Cref{subsec:recap}) makes this approximation algorithmic with near-optimal complexities $\tilde{O}(m)$ centralized and $\tilde{O}(D+\sqrt{n})$ distributed.

\subsubsection{A Review of Centralized Connectivity Decompositions}
\paragraph{Edge Connectivity} Edge connectivity decompositions into spanning tree packings of $\ceil{\frac{\lambda-1}{2}}$ have been known due to results of Tutte\cite{Tutte} and NashWilliams\cite{Nash-Williams} from 1960, and they have found many important applications: the best known centralized minimum edge cut algorithm~\cite{KargerSTOC96}, the network coding advantage in edge-capactitated networks~\cite{Li-Li-Lau}, and tight analysis of the number of minimum cuts of a graph and random edge-sampling~\cite{KargerSTOC96}. Centralized algorithms for finding such a spanning tree packing include: an $\tilde{O}(\min\{mn, \frac{m^2}{\sqrt{n}}\})$ time algorithm for unweighted graphs by Gabow and Westermann~\cite{GabowWestermann}, an $\tilde{O}(mn)$ time algorithm for weighted graphs by Barahona~\cite{Barahona1995packing}, and an $\tilde{O}(m\lambda)$ time algorithm for a fractional packing via the general technique of Plotkin et al.~\cite{Plotkin:1991} (see \cite{KargerSTOC96}).

\smallskip
\paragraph{Vertex Connectivity} As mentioned before, the case of vertex connectivity decompositions was addressed only recently~\cite{CGK14}, and it was shown to have applications on analyzing vertex connectivity under vertex sampling and also network coding gap in node-capacitated networks. Consequent to (a preliminary version of)~\cite{CGK14}, Ene et al.\cite{Alina} presented a nice alternative proof for obtaining fractional dominating tree packing of size $\Omega(\frac{k}{\log n})$, which uses metarounding results of Carr and Vempala~\cite{CarrVempala} and the Min-Cost-CDS approximation result of Guha and Khuller~\cite{GuhaKhuller}. That proof does not extend to integral packing. Even though the proofs presented in \cite{CGK14} and \cite{Alina} are based on polynomial time algorithms, neither of the algorithms seems to admit a distributed implementation and even their centralized complexities are at least $\Omega(n^3)$.

\section{Notations and Problem Statements}\label{sec:prelim}
Given an undirected graph $G = (V, E)$ and a set $S\subseteq V$, we use the notation $G[S]$ to indicate the subgraph of $G$ induced by $S$. A set $S\subseteq V$ is called a \emph{connected
    dominating set (CDS)} iff $G[S]$ is connected and for each node $u \in V\setminus S$, $u$ has a neighbor $v\in S$. A subgraph $T=(V_{T}, E_{T})$ of graph $G=(V_{G}, E_{G})$ is a \emph{dominating tree} of graph $G$ if $T$ is a tree and $V_{T}$ is a dominating subset of $V_{G}$.


%

\paragraph{Dominating Tree Packing}
Let $DT(G)$ be the set of all dominating trees of $G$. A $\kappa$-size dominating tree packing of $G$ is a collection of $\kappa$ \emph{vertex-disjoint} \emph{dominating} trees in $G$.
A $\kappa$-size fractional dominating tree packing of $G$ assigns a weight $x_{\tau} \in [0,1]$ to each $\tau \in DT(G)$
such that $\sum_{\tau \in DT(G)} x_{\tau} = \kappa$ and $\forall v \in V$, $\sum_{\tau, v\in\tau}x_\tau\leq 1$.

\paragraph{Spanning Tree Packing}
Let $ST(G)$ be the set of all spanning trees of $G$. A $\kappa$-size dominating tree packing of $G$ is a collection of $\kappa$ \emph{edge-disjoint} \emph{spanning} trees in $G$.
A $\kappa$-size fractional spanning tree packing of $G$ assigns a weight $x_{\tau} \in [0,1]$ to each $\tau \in ST(G)$
such that $\sum_{\tau \in ST(G)} x_{\tau} = \kappa$ and $\forall e \in E$, $\sum_{\tau, e\in\tau}x_\tau\leq 1$.

\paragraph{Distributed Problem Requirements}
In the distributed versions of dominating or spanning tree packing problems, we consider each tree $\tau$ as one class with a unique identifier $ID_{\tau}$ and a weight $x_{\tau}$. In the spanning tree packing problem, for each node $v$ and each edge $e$ incident to $v$,  for each spanning tree $\tau$ that contains $e$, node $v$ should know $ID_{\tau}$ and $x_{\tau}$. In the dominating tree packing problem, for each node $v$ and each dominating tree $\tau$ that contains $v$, node $v$ should know $ID_{\tau}$, $x_\tau$, and the edges of $\tau$ incident to $v$. 

\paragraph{Distributed Model Details}
See \Cref{subsec:results} for the definitions of our communication models. Note that as we work with randomized algorithms, nodes can generate random ids by each taking random binary strings of $4\log n$ bits and delivering it to their neighbors. Moreover, we assume no initial knowledge about the graph. Note that in our models, by using a simple and standard BFS tree approach, in $O(D)$ rounds, nodes can learn the number of nodes in the network $n$, and also a $2$-approximation of the diameter of the graph $D$, which is enough for our applications. Our algorithms assume this knowledge to be ready for them.

\section{Fractional Dominating Tree Packing Algorithm}
\label{sec:PackingAlg}
In this section, we present the main technical contribution of the paper, which is introducing a new algorithm for vertex connectivity decomposition that has near optimal time complexities for both centralized and distributed implementations. In this decomposition, we construct a collection of $\Omega(k)$ classes, each of which is a dominating tree w.h.p., such that each vertex is included in at most $O(\log n)$ classes. This gives a fractional dominating tree packing\footnote{The approach of this algorithm can be also used to get an $\Omega(\frac{\kappa}{\log^5 n})$ dominating tree packing, see \cite[Section 4]{CGK14}.} of size $\Omega(\frac{k}{\log n})$ and lets us achieve Theorems \ref{distDomPack} and \ref{centDomPack}. The analysis is presented in \Cref{sec:PackingAnalysis}.

\subsection{The Algorithm Outline}\label{subsec:outline}
For the construction, we first assume that we have a $2$-approximation of $k$, and then explain how to remove this assumption. 
 
We construct $t=\Theta(k)$ connected dominating sets (CDS) such that each node is included in $O(\log n)$ CDSs. We work with CDSs, since it is simply enough to determine their vertices. To get dominating trees, at the end of the CDS packing algorithm, we remove the cycles in each CDS using a simple application of a minimum spanning tree algorithm. 

We transform the graph $G=(V, E)$ into a graph $\mathcal{G}=(\mathcal{V}, \mathcal{E})$, which is called the \emph{virtual graph}~\footnote{The virtual graph $\mathcal{G}$ is nothing but using $\Theta(\log n)$ copies of $G$, or simply reusing each node of $G$ for $\Theta(\log n)$ times, $\Theta(1)$ times per \emph{layer} (described later). We find it more formal to use $\mathcal{G}$ instead of directly talking about $G$.}, and is constructed as follows: Each node $v \in V$ simulates $\Theta(\log n)$ \emph{virtual} nodes $\nu_{1}, \nu_2, \dots, \nu_{\Theta(\log n)} \in \mathcal{V}$ and two virtual nodes are connected if they are simulated by the same real node or by two $G$-adjacent real nodes. To get the promised CDS Packing, we partition the virtual nodes $\mathcal{V}$ into $t$ disjoint \emph{classes}, each of which is a CDS of $\mathcal{G}$, w.h.p. Each CDS $\mathcal{S}$ on $\mathcal{G}$ defines a CDS $S$ on $G$ in a natural way: $S$ includes all real nodes $v$ for which at least one virtual node of $v$ is in $\mathcal{S}$. Thus, the $t$ classes of virtual nodes w.h.p. give $t$ CDSs on $G$ and clearly each real node is included in $O(\log n)$ CDSs.

For the construction, we organize the virtual nodes by giving them two attributes: each virtual node has a \emph{layer number} in $\{1, 2, \dots, \NumOfLayers\}$, where $\NumOfLayers=\Theta(\log n)$, and a \emph{type number} in $\{1,2,3\}$. For each real node $v \in V$, the $3L=\Theta(\log n)$ virtual nodes simulated by $v$ are divided such that, for each layer number in $\{1, 2, \dots, \NumOfLayers\}$ and each \emph{type number} in $\{1,2,3\}$, there is exactly one virtual node. For the communication purposes in the distributed setting, note that each communication round on $\mathcal{G}$ can be simulated via $\Theta(\log n)$ communication rounds on $G$. Thus, to simplify discussions, we divide the rounds into groups of $\Theta(\log n)$ consecutive rounds and call each group one \emph{meta-round}. 

As explained, we assign each virtual node to a class. This class assignment is performed in a recursive manner based on the layer numbers. First, with a jump-start, we assign each virtual node of layers $1$ to $\NumOfLayers/2$ to a random class in classes $1$ to $t$. This step gives us that each class dominates $\mathcal{G}$, w.h.p. The interesting and challenging part is to achieve connectivity for all classes. For this purpose, we go over the layers one by one and for each layer $\ell \in [\NumOfLayers/2, \NumOfLayers-1]$, we assign class numbers to the virtual nodes of layer $\ell+1$ based on the
assignments to the virtual nodes of layers $1$ to $\ell$. The goal is to connect the components of each class such that the total number of connected components (summed up over all classes) decreases (in expectation) by a constant factor, with the addition of each layer. This would give us that after $\Theta(\log n)$ layers, all classes are connected, w.h.p. We next explain the outline of this step, after presenting some notations.

Let $\mathcal{V}^i_\ell$ be the set of virtual nodes of layers $1$ to $\ell$ assigned to class $i$ (note that $\mathcal{V}^i_\ell \subseteq \mathcal{V}^{i}_{\ell+1}$). Let $N^i_\ell$ be the number of
connected components of $\mathcal{G}[\mathcal{V}^i_\ell]$ and let $M_\ell:=\sum_{i=1}^{t} (N^i_\ell-1)$ be the total number of
excess components after considering layers $1, \dots, \ell$, compared to the ideal case where each class is connected. Initially $M_1\leq n-t$, and as soon as $M_\ell=0$, each class
induces a connected subgraph.

\paragraph{Recursive Class Assignment} Suppose that we are at the step of assigning classes to virtual nodes of layer $\ell+1$. We call virtual nodes of layer $\ell+1$ \emph{new nodes} and the virtual nodes of layers $1$ to $\ell$ are called \emph{old nodes}. Also, in the sequel, our focus is on the virtual nodes and thus, unless we specifically use the phrase ``real node", we are talking about a virtual node. First, each new node of type $1$ or type $3$ joins a random class. It then remains to assign classes to type-$2$ new nodes, which is the key part of the algorithm. The outline of this procedure is as follows:

\smallskip
\begin{mdframed}[hidealllines=true,backgroundcolor=gray!20]
\begin{center}
  \parbox{1.00\linewidth}{ \textbf{Recursive Class Assignment Outline: } 
    \begin{itemize}[noitemsep,topsep=5pt,parsep=6pt,partopsep=0pt] 
    \item[(1)] \textbf{\emph{Identify the connected components of old
        nodes}}, i.e., those of $\mathcal{G}[\mathcal{V}^i_\ell]$ for each class $i$.  
    \item[(2)] \textbf{\emph{Create the bridging graph}}, a bipartite graph between the connected components of old nodes
      and type-$2$ new nodes defined as follows: We view each connected component of old nodes as one node on one side of the bipartite graph, by assuming all its vertices are contracted into one node. Each type-$2$ new node $v$ is adjacent to a connected component $\mathcal{C}$ of $\mathcal{G}[\mathcal{V}^i_\ell]$ if all the following conditions hold: 
			(a) $v$ has a neighbor in $\mathcal{C}$, (b) $\mathcal{C}$ does not have a
      type-$1$ new node neighbor $u$ such that $u$ joined class $i$ and that $u$ has a
      neighbor in a component $\mathcal{C}' \neq \mathcal{C}$ of
      $\mathcal{G}[\mathcal{V}^i_\ell]$, and (c) $v$ has a type-$3$ new node neighbor $w$
      such that $w$ has joined class $i$ and $w$ has a neighbor in a
      connected component $\mathcal{C}'' \neq \mathcal{C}$ of
      $\mathcal{G}[\mathcal{V}^i_\ell]$. 
    \item[(3)] \textbf{\emph{Find a maximal matching $\mathcal{M}$ in
        the bridging graph}}, i.e., between components and type-$2$
      new nodes. For each type-$2$ new node $v$, if it is matched in
      $\mathcal{M}$, then it joins the class of its matched component and otherwise, it joins a random class.
    \end{itemize}
	\vspace{-0.2cm}
  }
\end{center}
\end{mdframed}
\medskip


%
Intuitively, the rule described in step (2) means the following: $v$ is a neighbor of $\mathcal{C}$ in the bridging graph if component $\mathcal{C}$ is not (already) connected to another component of $\mathcal{G}[\mathcal{V}^i_\ell]$ via one type-1 new node, but if $v$ joins class $i$, then with the help of $v$ and $w$, the component $\mathcal{C}$ will be merged with some other component $\mathcal{C}'' \neq \mathcal{C}$ of $\mathcal{G}[\mathcal{V}^i_\ell]$. See \Cref{fig:CCG}.

This recursive class assignment outline can be implemented in a distributed setting in $\tilde{O}(\min\{D+\sqrt{n}, \frac{n}{k}\})$ rounds of the \LBM model, and in a centralized setting in $\tilde{O}(m)$ steps, thus proving respectively \Cref{distDomPack} and \Cref{centDomPack}. We present the details of the distributed implementation in \Cref{subsec:DistImp}. The centralized implementation is presented in Appendix~\ref{sec:centPack}. 

\begin{figure}
\centering
\includegraphics[width=60 mm]{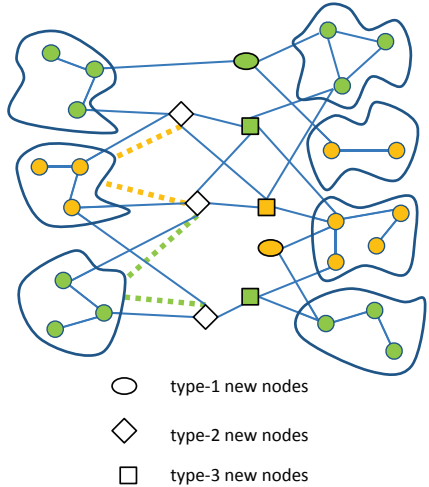}
\caption{Bridging graph edges are presented by dotted lines;
  each color is one class.}
\label{fig:CCG}
\vspace{-0.45cm}
\end{figure}
We show in the analysis that the above algorithm indeed constructs $k$ CDSs, w.h.p., and clearly each real node is contained in at most $O(\log n)$ CDSs, at most one for each of its virtual nodes. To turn these CDSs into dominating trees, we simply remove some of the edges of each class so as to make it a tree. We do this by a simple application of a minimum spanning tree algorithm on the virtual graph $\mathcal{G}$: We give weight of $0$ to the edges between virtual nodes of the same class and weight $1$ to other edges. Then, the edges with weight $0$ that are included in the minimum spanning tree of $\mathcal{G}$ form our dominating trees, exactly one for each CDS. 

\begin{remark}\label{rmrk:KnowledgeRemoval}
The assumption of knowing a 2-approximation of $k$ can be removed at the cost of at most an $O(\log n)$ increase in the time complexities.
\end{remark}
To remove the assumption, we use a classical try and error approach: we simply try exponentially decreasing guesses about $k$, in the form $\frac{n}{2^j}$, and we test the outcome of the dominating tree packing obtained for each guess (particularly its domination and connectivity) using a randomized testing algorithm presented in Appendix \ref{sec:App6} on the virtual graph $\mathcal{G}$. For this case, this test runs in a distributed setting with a round complexity of $O(\min\{\frac{n\log^2 n}{k}, (D + \sqrt{n\log n}\log^*n)\log^2 n\})$ in the \LBM model, and in a centralized setting with step complexity of $O(m\log^3 n)$.


\medskip
\paragraph{An intuitive comparison with the approach of \cite{CGK14}} We note that the approach of the above algorithm is significantly different than that of \cite{CGK14}. Mainly, the key part in \cite{CGK14} is that it finds short paths between connected components of the same class, called \emph{connector paths}. The high-level idea is that, by adding the nodes on the connector paths of a class to this class, we can merge the connected components at the endpoints of this path. However, unavoidably, each node of each path might be on connector paths of many classes. Thus, the class assignment of this node is not clear. \cite{CGK14} cleverly allocates the nodes on the connector paths to different classes so as to make sure that the number of connected components goes down by a constant factor in each layer (in expectation).

Finding the connector paths does not seem to admit an efficient distributed algorithm and it is also slow in a centralized setting. The algorithm presented in this paper does not find connector paths or use them explicitly. However, it is designed such that it enjoys the existence of connector paths and its performance gains implicitly from the abundance of the connector paths. While this is the key part that allows us to make the algorithm distributed and also makes it simpler and faster centralized, the analysis becomes more involved. The main challenging part in the analysis is to show that the the size of the maximal matching found in the bridging graph is large enough so that in each layer, the number of connected components (summed up over all classes) goes down by a constant factor, with at least a constant probability. This is addressed in \Cref{subsec:FML}.


\section{Dominating Tree Packing Analysis}
\label{sec:PackingAnalysis}
In this section, we present the analysis for the algorithm explained in \Cref{sec:PackingAlg}. We note that this analysis is regardless of whether we implement the algorithm in a distributed or a centralized setting. In a first simple step, we show that each class is a dominating set. Then, proving the connectivity of all classes, which is the core technical part, is divided into two subsections: We first present the concept of connector paths in \Cref{subsec:Connectors} and then use this concept in \Cref{subsec:FML} to achieve the key point of the connectivity analysis, i.e., the Fast Merger Lemma (\Cref{lem:FML3}). Some simpler proofs are deferred to \Cref{App:MissingAnalysis}.

\begin{lemma}[\textbf{Domination Lemma}] \label{lem:layer1dom} 
  W.h.p., for each class $i$, $\mathcal{V}^i_{\NumOfLayers/2}$ is a dominating
  set.
\end{lemma}

Note that since $\mathcal{V}^i_\ell \subseteq \mathcal{V}^{i}_{\ell'}$ for $\ell\leq \ell'$, the domination of each class follows directly from this
lemma. For the rest of this section, we assume that for each class
$i$, $\mathcal{V}^i_{\NumOfLayers/2}$ is a dominating set.

\subsection{Connector Paths}\label{subsec:Connectors}
The concept of connector paths is a simple toolbox that we developed in \cite{CGK14}. For completeness, we present a considerably simpler version here:

Consider a class $i$, suppose $N^i_{\ell} \geq 2$, and consider a
component $\mathcal{C}$ of $\mathcal{G}[\mathcal{V}^i_\ell]$. For each set of virtual vertices
$\mathcal{W} \subseteq \mathcal{V}$, define the projection
$\Psi(\mathcal{W})$ of $\mathcal{W}$ onto $G$ as the set $W \subseteq
V$ of real vertices $w$, for which at least one virtual node of $w$ is in
$\mathcal{W}$. 

A path $P$ in the real graph $G$ is called a \emph{potential
  connector} for $\mathcal{C}$ if it satisfies the following
three conditions: (A) $P$ has one endpoint in $\Psi(\mathcal{C})$ and
the other in $\Psi(\mathcal{V}^i_{\ell} \setminus
\mathcal{C})$, (B) $P$ has at most two internal vertices, (C) if $P$ has
exactly two internal vertices and has the form $s$, $u$, $w$, $t$ where
$s\in \Psi(\mathcal{C})$ and $t \in \Psi(\mathcal{V}^i_{\ell}
\setminus \mathcal{C})$, then $w$ does not have a neighbor in
$\Psi(\mathcal{C})$ and $u$ does not have a neighbor in
$\Psi(\mathcal{V}^i_{\ell} \setminus \mathcal{C})$. 

Intuitively, condition (C) requires \emph{minimality} of each potential connector path. That is, there is no potential connector path connecting $\Psi(\mathcal{C})$ to another component of $\Psi(\mathcal{V}^i_{\ell})$ via only $u$ or only $w$. 

From a potential connector path $P$ on graph $G$, we derive a
connector path $\mathcal{P}$ on virtual graph $\mathcal{G}$ by determining the types of the related internal virtual vertices as
follows: (D) If $P$ has one internal real vertex $w$, then for
$\mathcal{P}$ we choose the virtual vertex of $w$ in layer ${\ell}+1$ in
$\mathcal{G}$ with type $1$. (E) If $P$ has two internal real vertices
$w_1$ and $w_2$, where $w_1$ is adjacent to $\Psi(\mathcal {C})$ and
$w_2$ is adjacent to $\Psi(\mathcal{V}^i_{\ell} \setminus \mathcal
{C})$, then for $\mathcal{P}$ we choose the virtual vertices of $w_1$ and
$w_2$ in layer $\ell+1$ with types $2$ and $3$, respectively. Finally,
for each endpoint $w$ of $P$ we add the copy of $w$ in
$\mathcal{V}^i_{\ell}$ to $\mathcal{P}$. We call a connector path
that has one internal vertex a \emph{short connector path}, whereas a
connector path with two internal vertices is called a \emph{long connector path}. An example is demonstrated in \Cref{fig:ConnectorPaths}. 
Because of condition (C), and rules (D) and (E) above, we get the
following fact:
\begin{proposition}\label{prop:longpaths}
  For each class $i$, each type-$2$ virtual vertex $u$ of layer $\ell+1$
  is on connector paths of at most one connected component of
  $\mathcal{G}[\mathcal{V}^i_{\ell}]$.
\end{proposition}

We next show that each component that is not single in its
class has $k$ connector paths. 
\begin{lemma} [\textbf{Connector Abundance Lemma}] \label{lem:CAL}
  Consider a layer $\ell \geq \NumOfLayers/2+1$ and a class $i$ such
  that $\mathcal{V}_{L/2}^i\subseteq \mathcal{V}_\ell^i$ is a
  dominating set of $\mathcal{G}$ and $N^i_{\ell}\geq 2$. Further
  consider an arbitrary connected component $\mathcal{C}$ of
  $\mathcal{G}[\mathcal{V}^i_\ell]$. Then, $\mathcal{C}$ has at least $k$ internally vertex-disjoint connector paths.
\end{lemma}
The proof is based on a simple application of Menger's theorem ~\cite[Theorem 9.1]{bondymurty} and crucially uses the domination of each class, \Cref{lem:layer1dom}. The proof can also be viewed as a simplified version of that of \cite[Lemma 3.4]{CGK14}. To be self-contained, we present a proof in \Cref{App:MissingAnalysis}.

\subsection{The Fast Merger Lemma}\label{subsec:FML}
We next show that the described algorithms will make the number of connected components go down by a constant factor in each layer. The formal statement is as follows:

\begin{lemma} [\textbf{Fast Merger Lemma}] \label{lem:FML3} For each
  layer ${\ell} \in [\frac{L}{2}, L-1]$, $M_{\ell+1} \leq
  M_\ell$, and moreover, there are constants $\delta,\rho>0$ such that
  $\Pr[M_{{\ell}+1} \leq (1-\delta) \cdot M_{\ell}] \geq \rho$ with
  independence between layers.
\end{lemma}

%
\begin{proof}
 Let ${\ell}$ be a layer in $[\frac{L}{2}, L-1]$. For the first part note that since from \Cref{lem:layer1dom} we know that $\mathcal{V}^i_{\NumOfLayers/2}$ is a dominating, and as $\mathcal{V}^i_\ell \subseteq \mathcal{V}^{i}_{\ell'}$ for $\ell\leq \ell'$, each new node of layer $\ell+1$ that is added to class $i$ has a neighbor in the old components of this class and thus, the new nodes do not increase the number of connected components.

  For the second part, let $i$ be
  a class for which $N^i_{\ell} \geq 2$ and consider a component
  $\mathcal{C}$ of $\mathcal{G}[\mathcal{V}^i_{\ell}]$. We say that $\mathcal{C}$ is
  \emph{good} if one of the following two conditions is
  satisfied: 
	(I) There is a type-$1$ new node $v$ that has a neighbor in
    $\mathcal{C}$ and a neighbor in a component $\mathcal{C}'\neq
    \mathcal{C}$ of $G^i_{\ell}$ and $v$ joins class $i$. 
	(II) There are two neighboring new nodes, $w$ and $u$, with types $2$ and $3$, respectively, such that $w$ has a neighbor in
    $\mathcal{C}$, $u$ has a neighbor in a component $\mathcal{C}' \neq \mathcal{C}$ of $\mathcal{G}[\mathcal{V}^i_{\ell}]$, and both $u$ and $w$ join
  class $i$.
	Otherwise, we say that $\mathcal{C}$ is \emph{bad}.  By
  definition, if a connected component of old nodes is good, then at
  the next layer it is merged with another component of its class. 
	
	Let $X^i_{\ell}$ be the number of bad connected components of class $i$
  if $N^i_{\ell} \geq 2$ and $X^i_{\ell}=0$ otherwise. Also, 
  define $Y_{\ell}= \sum_{i=1}^{t} X^i_{\ell}$, which gives
  $M_{{\ell}+1} \leq \frac{M_{\ell} - Y_{\ell}}{2} + Y_{\ell}$. To prove the lemma, we show that $\E[Y_{\ell}] \leq
  (1-3\delta) \cdot M_{\ell}$ for some constant $\delta>0$. Then, using
  Markov's inequality we get that $\Pr[Y_{\ell} \leq (1-2\delta) \cdot
  M_{\ell}] \geq 1-\frac{1-3\delta}{1-2\delta}$ and therefore
  $\Pr\left[M_{\ell+1}\leq (1-\delta)M_{\ell}\right]\geq
  1-\frac{1-3\delta}{1-2\delta}$ and thus the lemma follows.

  Hence, it remains to prove that $\E[Y_{\ell}] \leq (1-3\delta)
  \cdot M_{\ell}$ for some $\delta>0$. For this, we divide the connected
  components of old nodes into two groups of \emph{fast} and
  \emph{slow} components, as follows: Consider a class $i$ such that
  $N^i_{\ell} \geq 2$. A connected component $\mathcal{C}$ of
  $\mathcal{G}[\mathcal{V}^i_{\ell}]$ is called \emph{fast} if it has at least
  $\Omega(k)$ short connector paths, and \emph{slow}
  otherwise. Note that by \Cref{lem:CAL}, w.h.p., each slow component
  has at least $\Omega(k)$ long connector paths.

  Let $M^F_{\ell}$ and $M^S_{\ell}$ be the total number of fast and
  slow connected components, respectively, of graphs $\mathcal{G}[\mathcal{V}^i_{\ell}]$ as
  $i$ ranges over all classes. Note that $M^F_{\ell} + M^S_{\ell} =
  M_{\ell}$. We say that a short connector path $p$ for $\mathcal{C}$
  is \emph{good} if its internal node is in the same class as
  $\mathcal{C}$. Let $Y^F_{\ell}$ be the total number of fast
  connected components for which none of the short paths is
  good. Because every type-1 new node picks its class number randomly, each of the
  $\Omega(k)$ short paths (independently) has
  probability at least $1/k$ to be in the right
  class. The expected number of short paths in the right class is
  therefore constant and hence, there exists a constant $\delta>0$ such
  that $\E[Y^F_{\ell}] \leq (1-3\delta) \cdot
  M^F_{\ell}$. 

  Moreover, let $\mathcal{K}$ be set of the slow connected components of
  the graphs $\mathcal{G}[\mathcal{V}^i_{\ell}]$ (for all classes $i\in [1,t]$) for which none of the
  short paths is good and let $K := |\mathcal{K}|$. Let
  $\mathcal{M}$ be the maximal matching the algorithm computes for the
  bridging graph. In order to complete the proof, we show that the
  expected size of $\mathcal{M}$ is at least $3\delta \cdot K$
  for some $\delta>0$. Given this, we get that
  \begin{eqnarray}
    \E[Y_{\ell}] 
     =  \E[Y^F_{\ell}] + \big(\E[K] - 
    \E[|\mathcal{M}|] \big) 
     \leq  (1-3\delta) \cdot (M^F_{\ell} + K) \leq (1-3\delta) \cdot (M^F_{\ell} + M^{S}_{\ell}) 
     =  (1-3\delta) \cdot M_{\ell}, \nonumber 
  \end{eqnarray}
  which would complete the proof. 

  To show that the expected size of the maximal matching is at least
  $3\delta \cdot K$ for some $\delta>0$, it is sufficient to
  prove that the expected size of a maximum matching is at least
  $\Omega(K)$. It is well-known and easy to see that the size of any
  maximal matching is at least half of the size of a maximum matching. Hence, what is left of the proof is to show that the expected size of the maximum matching is at least $\Omega(K)$. We do this in \Cref{lem:maxMatchSize}.
	\end{proof}
	
	\begin{lemma}\label{lem:maxMatchSize}The expected size of the maximum matching in the bridging graph is at least $\Omega(K)$.
	\end{lemma}
This lemma is the key part of the analysis. As the proof has many technical details, we defer it to \Cref{App:MissingAnalysis} and only mention a rough outline here: The bridging graph might have a complex structure and thus, we do not know how to work with it directly. However, the saving grace is that we know more about the connector paths, thanks to \Cref{lem:CAL}. Using long connector paths, we algorithmically identify a (random) subgraph $\mathcal{H}$ of the bridging graph and show that just this subgraph $\mathcal{H}$ contains a matching of size $\Omega(K)$. The analysis of this algorithm uses a simple probability tail bound inequality that we develop for our specific problem.

\begin{lemma}\label{lem:size} W.h.p., for each $i$, the number of
  virtual nodes in class $i$ is $O(\frac{n\log n}{k})$. 
\end{lemma}


\section{Distributed Fractional Spanning-Tree Packing}
Recall that the celebrated results of Tutte~\cite{Tutte} and Nash-Williams~\cite{Nash-Williams} show that each graph with edge connectivity $\lambda$ contains $\ceil{\frac{\lambda-1}{2}}$ edge-disjoint spanning trees. In this section, we prove \Cref{distSpanPack} which achieves a fractional spanning tree packing with almost the same size. In \Cref{subsec:STsmall}, we explain the algorithm for the case where $\lambda = O(\log n)$. We later explain in \Cref{subsec:STlarge} how to extend this algorithm to the general case. In the interest of space, we defer the analysis to \Cref{app:STpack}. 

\subsection{Fractional Spanning Tree Packing for $\lambda = O(\log n)$}\label{subsec:STsmall}
We follow a classical and generic approach (see e.g.~\cite{Plotkin:1991, shahrokhi1990maximum, klein1994faster, KargerSTOC96}) which can be viewed as an adaptation of the Lagrangian relaxation method of optimization theory. Tailored to our problem, this approach means we always maintain a collection of weighted trees which might have a large weight going through one edge, but we iteratively improve this collection by penalizing the edges with large load, which incentivizes the collection to take some weight away from the edges with larger load and distribute it over the edges with smaller loads. We next present the formal realization of this idea. 

\medskip
\paragraph{Algorithm Outline} We will always maintain a collection $T$ of weighted trees---where each tree $\tau \in T$ has weight $w_{\tau} \in [0,1]$---such that the total weight of the trees in the collection is $1$. That is $\sum_{\tau\in T} w_{\tau}=1$. We start with a collection containing only one (arbitrary) tree with initial weight $1$ and iteratively improve this collection for $\Theta(\log^3 n)$ iterations: During each iteration, for each edge $e\in E$, let $x_e$ be the \emph{weighted load} on edge $e$, that is $x_{e} = \sum_{\tau, e \in \tau} w_{\tau}$ and also, let $z_e = x_e \ceil{\frac{\lambda-1}{2}}$. Our goal is that at the end, we have $\max_{e\in E} z_e \leq 1+O(\eps)$. 

In each iteration, for each edge $e$, we define a cost $c_e= exp(\alpha \cdot z_e)$, where $\alpha=\Theta(\log n)$. Then, we find the Minimum Spanning Tree (MST) with respect to these costs. If $Cost(MST)=\sum_{e\in MST} c_e > (1-\eps) \sum_{e\in E} c_e \cdot x_e$, then the algorithm terminates. On the other hand, if $Cost(MST)=\sum_{e\in MST} c_e \leq (1-\eps) \sum_{e\in E} c_e \cdot x_e$, then we add this MST to our weighted tree collection $T$, with weight $\beta=\Theta(\frac{1}{\alpha\log n})$, and to maintain condition $\sum_{\tau\in T} w_{\tau}=1$, we multiplying the weight of the old trees in $T$ by $1-\beta$.

\paragraph{Distributed Implementation} Using the beautiful distributed minimum spanning tree algorithm of Kutten and Peleg~\cite{KuttenPeleg95}, we can perform one iteration in $O(D+\sqrt{n}\log^*n)$ rounds of the \LBM model~\footnote{For communication purposes, \cite{KuttenPeleg95} assumes that the weight of each edge can be described in $O(\log n)$ bits. In our algorithm, the weight of each edge $e$ is in the form  $c_e= exp(\alpha \cdot z_e)$ and can be potentially super-polynomial, which means the naive way of sending it would require $\omega(\log n)$ bits. However, it simply is enough to send $z_e$ instead of $c_e$ as then the receiving side can compute $c_e$. Fortunately, the maximum value that $z_e$ can obtain is $\Theta(\log^3 n)$ and we can always round it to multiples of e.g. $\Theta(\frac{1}{n})$ with negligible $o(1)$ effect on the collection's final load on each edge.}. Hence, the $\Theta(\log^3 n)$ iterations of the above algorithm can be performed in $O((D+\sqrt{n}\log^*n)\log^3 n)$ rounds. Note that in these iterations, each node $v$ simply needs to know the weight on edges incident on $v$ and whether another iteration will be used or not. The latter decision can be made centrally---in a leader node, e.g., the node with the largest id---by gathering the total cost of the minimum spanning tree over a breadth first search tree rooted at this leader and then propagating the decision of whether to continue to next iteration or not to all nodes.

\subsection{Generalized fractional Spanning Tree Packing}\label{subsec:STlarge}
The key idea for addressing the general case of $\lambda$---specially when $\lambda=\Omega(\log n)$---is that we randomly decompose the graph into spanning subgraphs each with connectivity $\min\{\lambda, \Theta(\log n/\eps^2)\}$ using random edge-sampling and then we run the edge-connectivity decomposition in each subgraph. 

The famous random edge-sampling technique of Karger~\cite[Theorem 2.1]{KargerSTOC94} gives us that, if we randomly put each edge of the graph in one of $\eta$ subgraphs $H_1$ to $H_\eta$, where $\eta$ is such that $\frac{\lambda}{\eta} \geq \frac{10\log n}{\eps^2}$, then each subgraph has edge-connectivity in $[\frac{\lambda}{\eta}(1-\eps), \frac{\lambda}{\eta}(1+\eps)]$ with high probability. Having this, we first find a $3$-approximation $\tilde{\lambda}$ of $\lambda$ using the distributed minimum edge cut approximation of Ghaffari and Kuhn~\cite{GK13} in ${O}((D+\sqrt{n}\log^*n)\log^2 n \log\log n)$ rounds. Then, using $\tilde{\lambda}$, we choose $\eta$ such that we are sure that $\frac{\lambda}{\eta} \in [\frac{20\log n}{\eps^2}, \frac{60\log n}{\eps^2}]$ and then put each edge in a random subgraph $H_1$ to $H_\eta$. This way, each subgraph has edge-connectivity in range $[\frac{10\log n}{\eps^2}, \frac{100\log n}{\eps^2}]$, which as $\eps$ is a constant, fits the setting of \Cref{subsec:STsmall}. On the other hand, the summation of the edge-connectivities $\lambda_1$ to $\lambda_\eta$ of subgraphs $H_1$ to $H_\eta$ is at least $\lambda(1-\eps)$. 

The remaining problem is to solve the spanning tree packing problem in each subgraph $H_i$, all in parallel. Recall that the algorithm explained in \Cref{subsec:STsmall} for the case of $O(\log n)$ edge connectivity, requires solving $O(\log^3 n)$ minimum spanning tree problems. Hence, if we solve the MSTs of different subgraphs naively with repetitive black-box usage of the MST algorithm of Kutten and Peleg~\cite{KuttenPeleg95}, this would take $O((D+\sqrt{n}\log^*n)\lambda \log^2 n)$ rounds. To obtain the round complexity $\tilde{O}(D+\sqrt{n\lambda})$, instead of a simple black-box usage, we do a few simple modifications. As these modifications require a brief review of \cite{KuttenPeleg95}, we defer the details to the proof of \Cref{lem:STpackingMST} in \Cref{app:STpack}.

\begin{lemma}\label{lem:STpackingMST} The fractional spanning tree packing of all subgraphs can be implemented simultaneously, all in $O((D+ \sqrt{\frac{n\lambda}{\log n}}\log^*n)\log^{3} n)$ rounds of the \CM model.
\end{lemma}

\section{Open Problems}
%
In this paper, we presented algorithms that achieve a fractional dominating tree packing of size $\Omega(\frac{k}{\log n})$. This $O(\log n)$ loss is unavoidable if we want to decompose the graph into dominating trees~\cite{CGK14}. However, it remains open whether one can decompose the graph into subgraphs which each provide vertex-connectivity $O(\log n)$, while preserving the total vertex-connectivity up to a constant factor. Note that this would allow one to achieve a throughput of $\Theta(k)$ messages per round using network coding in each of the subgraphs. 

For both the vertex connectivity decomposition and the edge connectivity decomposition, the round complexities of our distributed algorithms have a gap from the corresponding lower bounds that depends on the connectivity. It is unclear whether the upper or the lower bound is to be improved. 

Finally, it is interesting to see whether the ideas set forward in this paper regarding approximating vertex connectivity using connectivity decompositions (see \Cref{subsubsec:VCapprox})---specifically fractional dominating tree packings---would allow one to get closer to the Aho, Hopcroft, Ullman conjecture of an $O(m)$ time centralized vertex-connectivity computation.

\section{Acknowledgement}
We thank Noga Alon for bringing to our attention the resemblance between dominating tree packing and vertex independent trees (see \Cref{subsubsec:IndependentTrees}).

\pagebreak
{\small
\bibliographystyle{abbrv}
\bibliography{ref}
}

\appendix
\section{Gossiping: An Application Example for Decompositions}\label{app:example}
To illustrate the power of our connectivity decompositions in information dissemination, let us consider a simple and crisp example, the classical \emph{gossiping} problem (aka \emph{all-to-all broadcast}): Each node of the network has one $O(\log n)$ bits message and the goal is for each node to receive all the messages.

In the following, we study this problem in the \LBM model. We first explain the approach for a particular value of connectivity and then explain how this extends to any connectivity size.

If the network is merely connected, solving the problem in $O(n)$ rounds is trivial. Now suppose that the network has in fact a vertex connectivity of $\sqrt{n}$. Despite this extremely good connectivity, prior to this work, the aforementioned $O(n)$ rounds solution remained the best known bound. The main difficulty is that, even though we know that each vertex cut of the network admits a flow of $\sqrt{n}$ messages per round, it is not clear how to organize the transmissions such that a flow of $\Omega(\sqrt{n})$ (distinct) messages per round passes through each cut. Note that a graph with vertex connectivity $k$ can have up to $\Theta(2^{k}\cdot (\frac{n}{k})^2)$ vertex-cuts of size $k$~\cite{Kanevsky}.

Our vertex connectivity decomposition runs in $\tilde{O}(\sqrt{n})$ rounds in this example (as here $D=O(\sqrt{n})$) and constructs $O(\sqrt{n})$ dominating trees, each of diameter $\tilde{O}(n)$, where each node is contained in $O(\log n)$ trees. Then, to use this decomposition for gossiping, we do as follows: first each node gives its message to (one node in) a random one of the trees. Note that this is easy as each node has neighbors in all of the trees and it can easily learn the ids of those trees in just one round. Then, w.h.p, we have $O(\sqrt{n})$ messages in each tree, ready to be broadcast. We can broadcast all the messages inside each dominating trees in $\tilde{O}(\sqrt{n})$ rounds. Furthermore, by a small change, we can make sure that each node trasmits the messages assigned to its dominating trees and thus, each node in the network receives all the messages. Overall, this method solves the problem in $\tilde{O}(n)$ rounds.

We now state how this bound generalizes:

\begin{corollary}\label{crl:broadcast} Suppose that there are $N$ messages in arbitrary nodes of the network such that each node has at most $\eta$ messages. Using our vertex connectivity decomposition, we can broadcast all messages to all nodes in $\tilde{O}(\eta+\frac{N+n}{k})$ rounds of the \LBM model.
\end{corollary}
\begin{proof}
This approach is exactly as explained above. Each node gives its messages to random dominating trees of the decomposition and then we broadcast each message only using the nodes in its designated dominating tree. The vertex connectivity decomposition runs in $\tilde{O}(\min\{D+\sqrt{n}, \frac{n}{k}\})$. Then, delivering messages to the trees takes at most $\eta$ rounds. Finally, broadcasting messages using their designated trees takes $\tilde{O}(\frac{n}{k}+\frac{N}{k})$ rounds because the diameter of each tree is $O(\frac{n\log n}{k})$ and each tree w.h.p is responsible for broadcasting at most $O(\frac{N}{k}+\log n)$ messages.
\end{proof}
Note that the bound in \Cref{crl:broadcast} is optimal, modulo logarithmic factors, because, (1) $\frac{N}{k}$ is a clear information theoretic lower bound as per round only $O(k \log n)$ bits can cross each vertex cut of size $k$, (2) similarly, if a node has $\eta$ messages, it takes at least $\eta$ rounds to send them, (3) in graphs with vertex connectivity $k$, the diameter can be up to $\frac{n}{k}$. In fact, the diameter of the original graph is a measure that is rather irrelevant as, even if the diameter is smaller, achieving a flow of size $\tilde{\Theta}(k)$ messages per round unavoidably requires routing messages along routes that are longer than the shortest path.

\section{Distributed Implementation of Dominating Tree Packing}\label{subsec:DistImp}
Here we present the distributed implementation of the outline of \Cref{subsec:outline}:

\begin{theorem}\label{lem:distImp}
There is distributed implementation of the fractional dominating tree packing algorithm of \Cref{subsec:outline} in $O(\min\{\frac{n\log n}{k}, D+\sqrt{n\log n}\log^{*}n\} \log^3 n)$ rounds of the \LBM model.
\end{theorem}

Throughout the implementation, we make frequent use of the following protocol, which is an easy extension of the connected component identification algorithm of Thurimella~\cite[Algorithm 5]{Thurimella95}, which itself is based on the MST algorithm of Kutten and Peleg~\cite{KuttenPeleg95}:
	
\begin{theorem}\label{thm:ComponentID}\footnote{We note that the MST algorithm of \cite{KuttenPeleg95} and the component identification algorithm of \cite{Thurimella95} were originally expressed in the $\mathcal{CONGEST}$ model but it is easy to check that the algorithms actually work in the more restricted \LBM model. In particular, \cite[Algorithm 5]{Thurimella95} finds the smallest preorder rank $x_v$ (for node $v$) in the connected component of each node $u$ but the same scheme works with any other inputs $y_{v} \in \{0,1\}^{O(\log n)}$ such that $y_{v}\neq y_u$ for $v\neq u$. To satisfy this condition, we simply set $y_v=(x_v, id_v)$, that is, we append the id of each node to its variable. As a side note, we remark that computing a preorder in the \LBM model requires $\Omega(n)$ rounds but we never use that.}Suppose that we are given a connected network $G=(V,E)$ with $n$ nodes and diameter $D$ and a subgraph $G_{sub}=(V, E')$ where $E' \subseteq E$ and each network node knows the edges incident on $v$ in graphs $G$ and $G_{sub}$. Moreover, suppose that each network node has a value $x_{v} \in \{0,1\}^{O(\log n)}$. There is a distributed algorithm in the \LBM model with round complexity of $O(\min\{D', D+\sqrt{n}\log^{*}n\})$ which lets each node $v$ know the smallest (or largest) $x_u$ for nodes $u$ that are in the connected component of $G_{sub}$ that contains $v$. Here, $D'$ is the largest strong diameter amongst connected components of $G_{sub}$.
\end{theorem}

\subsection{Identifying the Connected Components of Old Nodes}
To identify connected components of old nodes, each old node---i.e., those in layers $1$ to $\ell$---sends a message to all its $\mathcal{G}$-neighbors declaring its class number. We put each $\mathcal{G}$-edge that connects two virtual nodes of the same class in a subgraph $\mathcal{G}_{old}$. Moreover, each virtual node $\nu$ sets it id $ID_{\nu}= (ID_v, layer_\nu, type_\nu)$ where $v$ is the real node that contains $\nu$. Then, using \Cref{thm:ComponentID}, each virtual old node learns the smallest ID in its $\mathcal{G}_{old}$-component and remembers this id as its $componentID_{\nu}$. Running the algorithm of \Cref{thm:ComponentID} takes $O(\min\{\frac{n\log n}{k}, D+\sqrt{n\log n}\log^{*}n\})$ meta-rounds as the diameter of $\mathcal{G}$ is $D=diam(G)$ and since each component of $\mathcal{G}_{old}$ contains at most $O(\frac{n\log n}{k})$ nodes (\Cref{lem:size}) and thus has strong diameter $O(\frac{n\log n}{k})$.

\subsection{Creating the Bridging Graph}\label{subsubsec:OldGraph}
In the bridging graph, a type-$2$ new
node $v$ is connected to a connected component $\mathcal{C}$ of
$\mathcal{G}[\mathcal{V}^i_\ell]$ if and only if the following condition holds: if $v$ joins class $i$,
then component $\mathcal{C}$ becomes connected to some other
component $\mathcal{C}'$ of $\mathcal{G}[\mathcal{V}^i_\ell]$ through a type-$3$ new node,
and $\mathcal{C}$ is not already connected to another component by type-$1$ new nodes.

%

We first deactivate the components of old nodes that are connected to
another component of the same class by type-$1$ new nodes which chose
the same class. This is because in this layer we do not need to spend
a type-$2$ new node to connect these components to other components of their class. To
find components that are already connected through type-$1$ new nodes,
first, every old node $v$ sends its class number and $\componentID_\nu$
to all its neighbors. Let $u$ be a type-$1$ new node that has joined
class $i$. If $u$ receives component IDs of two or more components of
class $i$, then $u$ sends a message containing $i$ and a special
connector symbol \textsc{``connector''} to its neighbors. Each
component of class $i$ that receives a message from a type-$1$ new node containing class $i$ and the special connector symbol gets deactivates, that is, the node sets it local variable $activity=false$. This deactivation decision can be disseminated inside components of $\mathcal{G}_{old}$ in $O(\min\{\frac{n\log n}{k}, D+\sqrt{n\log n}\log^{*}n\})$ meta-rounds using \Cref{thm:ComponentID}.

Now, we start forming the bridging graph. Each old node $v$ (even if $v$ is in a deactivated component) sends its
$\componentID_v$ and its $activity$ status to all neighbors. For a type-$3$ new node $w$, let
$C_w$ be the set of component IDs $w$ receives in this meta-round. Assume
that $w$ joined class $i$. The node $w$ creates a message $m_w$ using
the following rule: If $C_w$ does not contain the component ID of a
component of class $i$, then the message $m_w$ is empty. If $C_w$
contains exactly one component ID of class $i$, $m_w$
contains the class number $i$ and this component ID. Finally, if $C_w$
contains at least two component IDs of class $i$, $m_w$
contains the class number $i$ and a special indicator symbol
\textsc{``connector''}. We use this symbol instead of the full list of component IDs, due to message size considerations. Each type-$3$ new node $w$ sends $m_w$ to all its neighbors.

To form the bridging graph, each type-$2$ new node $v$ creates
a \emph{neighbors list} $\nlist_v$ of active components which are its
neighbors in the bridging graph, as follows:
Consider a component $\mathcal{C} \in
\mathcal{G}[\mathcal{V}^i_\ell]$. Node $v$ adds $\mathcal{C}$ to $\nlist_v$ if $v$ has a
neighbor in active component $\mathcal{C}$ and $v$ received a message $m_w$
from a type-$3$ new neighbor such that $m_w$ is for class number $i$
and $m_w$ either contains the ID of a component $\mathcal{C}' \neq
\mathcal{C}$ of $\mathcal{G}[\mathcal{V}^i_\ell]$, or $m_w$ contains the special
\textsc{``connector''} symbol.

\subsection{Maximal Matching in the Bridging Graph}
%
To select a maximal matching in the bridging graph, we simulate Luby's well-known
distributed maximal independent set algorithm~\cite{alon86,luby86}. Applied to computing a maximal
matching of a graph $H$, the variant of the algorithm we use works as
follows: The algorithm runs in $O(\log |H|)$ phases. Initially all edges of $H$ are
active. In each phase, each edge picks a random number from a large
enough domain such that the numbers picked by edges are distinct with at least a constant probability. An
edge that picks a number larger than all adjacent edges joins the
matching. Then, matching edges and their adjacent edges become inactive. \cite{alon86,luby86} showed that this
algorithm produces a maximal matching in $O(\log |H|)$ phases.

We adapt this approach to our case as follows. We have $O(\log n)$
stages, one for each phase of Luby's algorithm. Throughout these
stages, each type-$2$ new node $v$ that is still unmatched keeps track
of the active components that are still available for being matched to
it. This can be done by updating the neighbors list $\nlist_v$ to
the remaining matching options.
In each stage, each unmatched type-$2$ new node $v$
chooses a random value of $\Theta(\log n)$ bits for each component in
$\nlist_v$. Then, $v$ picks the component $\mathcal{C} \in \nlist_v$
with the largest random value and proposes a matching to $\mathcal{C}$
by sending a proposal message $m_v$ that contains (a) the ID of $v$,
(b) the component ID of $\mathcal{C}$ and (c) the random value chosen
for $\mathcal{C}$ by $v$.

Nodes inside an active connected component may receive a number of proposals
and their goal is to select the type-$2$ new node which proposed the
largest random value to any node of this component.
Each old node $u$ has a variable named $\acceptedProposal_u$, which is
initialized to the proposal received by $u$ with the largest
random value (if any). We use algorithm of \Cref{thm:ComponentID} with subgraph $\mathcal{G}_{old}$ (described in \Cref{subsubsec:OldGraph}) and with initial value $x_u$ of  each node $u$ being its $acceptedProposal_u$. Hence, in $O(\min\{\frac{n\log n}{k},D+\sqrt{n\log n}\log^*n)$ meta-rounds, each old node $u$ learns the largest $acceptedProposal_v$ amongst nodes $v$ which are in the same $\mathcal{G}_{old}$component as $u$. Then, $u$ sets its $acceptedProposal_u$ equal to this largest proposal and also sends this $\acceptedProposal_u$ to all its neighbors. If a type-$2$ new
node $v$ has its proposal accepted, then $v$ joins the class of that component. Otherwise, $v$ remains unmatched at this stage and updates
its neighbors list $\nlist_v$ by removing the components in $\nlist_v$ that accepted proposals of other type-$2$ new nodes (those
from which $v$ received an $\acceptedProposal$ message). This process is repeated for $O(\log n)$ stages. Each type-$2$ new node that
remains unmatched after these stages joins a random class.

\subsection{Wrap Up}
Now that we have explained the implementation details of each of the steps of the recursive class assignment, we get back to concluding the proof of \Cref{lem:distImp}.

\begin{proof}[Proof of \Cref{lem:distImp}]
  Since in the matching part of the algorithm each component accepts at most one proposal from a type-$2$
  new node, the described algorithm computes a matching between type-$2$ new nodes and components.
  From the analysis of Luby's algorithm \cite{alon86,luby86}, it
  follows that after $O(\log n)$ stages, the selected matching is maximal w.h.p. Note that in some cases, the described
  algorithm might match a type-$2$ node $v$ and a component
  $\mathcal{C}$ even if the corresponding edge in the bridging graph does not get the
  maximal random value among all the edges of $v$ and $\mathcal{C}$ in
  the bridging graph. However, it is straightforward to see in the analysis of \cite{alon86, luby86} that this can only speed up the
  process.
	
 Regarding the time complexity, for each layer, the identification of the connected components on the old nodes and also creating the bridging graph take  $O(\min\{\frac{n\log n}{k}, D+\sqrt{n}\log^* n\})$ meta-rounds. Then, for each layer we have $O(\log n)$
  stages and each stage is implemented in $O(\min\{\frac{n\log n}{k}, D+\sqrt{n\log n}\log^* n\})$ meta-rounds. Thus, the time complexity for the
  each layer is $O(\min\{\frac{n\log n}{k}, D+\sqrt{n\log n}\log^{*}n\} \log^2 n)$ rounds, which accumulates to $O(\min\{\frac{n\log n}{k}, D+\sqrt{n\log n}\log^{*}n\} \log^3 n)$ rounds over $\NumOfLayers = \Theta(\log n)$ layers.
	
Furthermore, at the end of the CDS packing construction, in order to turn the CDSs into dominating trees, we simply use a linear time minimum spanning tree algorithm of Kutten and Peleg~\cite{KuttenPeleg95} on virtual graph $\mathcal{G}$ with weight $0$ for edges between nodes of the same class and weight $1$ for other edges. Then, the $0$-weight edges included in the MST identify our dominating trees. Running the MST algorithm of~\cite{KuttenPeleg95} on virtual graph $\mathcal{G}$ takes at most $O(D+\sqrt{n\log n}\log^*n)$ meta-rounds, or simply $O((D+\sqrt{n\log n}\log^*n)\log n)$ rounds. This MST can be performed also in $O(\frac{n\log n}{k})$ meta-rounds just by solving the problem of each class inside its own graph, which has diameter at most $O(\frac{n\log n}{k})$. In either case, both of these time complexities are subsumed by the other parts.
\end{proof}

\section{Centralized Implementation of Dominating Tree Packing}\label{sec:centPack}
In this section, we explain the details of a centralized implementation of the CDS-Packing algorithm presented in \Cref{sec:PackingAlg}.

\begin{theorem}\label{lem:centImp}
There is centralized implementation of the fractional dominating tree packing algorithm of \Cref{subsec:outline} with time complexity of $O(m\log^2 n)$.
\end{theorem}

\begin{proof}
We use \emph{disjoint-set data structures} for keeping track of the connected components of the graphs of different classes. Initially, we have one set for each virtual node, and as the algorithm continues, we union some of these sets. We use a simple version of this data structure that takes $O(1)$ steps for find operations and $O(\eta \log {\eta})$ steps for union, where $\eta$ is the number of elements. Moreover, for each layer $\ell$ and each type $r$, we have one linked list which keeps the list of virtual nodes of layer $\ell$ and type $r$.

We start with going over real nodes and choosing the layer numbers and type numbers of their virtual nodes. Simultaneously, we also add these virtual nodes to their respective linked list, the linked list related to their layer number and type. This part takes $O(n\log n)$ time in total over all virtual nodes.

To keep the union-set data structures up to date, at the end of the class assignment of each layer $\ell$, we go over the edges of virtual nodes of this layer---by going over the virtual nodes of the related linked lists, and their edges---and we update the disjoint-set data structures. That is, for each virtual node $v$ in these linked lists, we check all the edges of $v$. If the other end of the edge, say $u$, has the same class as $v$, then we union the disjoint-set data structures of $v$ and $u$. Since throughout these steps over all layers, each edge of the virtual graph is checked for union at most twice---once from each side---there are at most $O(m\log^2 n)$ checking steps for union operations. Moreover, the cost of all union operations summed up over all layers is at most $O(n \log^2 n)$. Since $m \geq \frac{n k}{2}$ and $k = \Omega(\log n)$, the cost of unions is dominated by the $O(m\log^2 n)$ cost of checking.

Now we study the recursive class assignment process and its step complexity. For the base case of layers $1$ to $\NumOfLayers/2$, we go over the linked lists related to layers $1$ to $\NumOfLayers/2$, one by one, and set the class number for each virtual node in these lists randomly.

In the recursive step, for each layer $j+1$ we do as follows: We first go over the linked list of type-$1$ virtual nodes of layer $j+1$ and the linked list of type-$3$ virtual nodes of layer $j+1$ and for each node $v$ in these lists, we select the class number of $v$ randomly. Over all layers, these operations take time $O(n \log n)$. Now we get to the more interesting part, choosing the class numbers of type-$2$ nodes of layer $j+1$. Recall that this is done via finding a maximal matching in the bridging graph.

We first go over the linked list of type-$1$ virtual nodes of layer $j+1$ and for each node $v$ in this list, we do as follows. Suppose that $v$ has joined class $i$. We go over edges of $v$ and find the number of connected components of class $i$ that are adjacent to $v$. Then if this number is greater than or equal to two, we go over those connected components and mark them as \emph{deactivated} for matching.

Next, for each type-$2$ virtual node of layer $j+1$, we have one array of size $\NumOfLayers$, called \emph{potential-matches} array. Each entry of this array keeps a linked list of component ids. We moreover assume that we can read the size of this linked list in $O(1)$ time. Note that this can be easily implemented by having a length variable for each linked list.

Now we begin the matching process. For this, we start by going over the linked list of type-$3$ virtual nodes of layer $j+1$. For each virtual node $u$ in this list, we go over the edges of $u$ and do as follows: if there is a neighbor $w$ of $u$ which is in a layer in $[1, j]$ and is in the same class as $u$, then $u$ remembers the component id of $w$. This component id is obtained by performing a find operation on the disjoint-set data structure of $w$. After going over all edges, $u$ has a list of neighboring connected component ids of the same class as $u$. Let us call these \emph{suitable components} for $u$. Then, we go over all the edge of $u$ for one more time and for each type-$2$ virtual neighbor $v$ of $u$ that is in layer $j+1$, we add the list of suitable components of $u$ into the entry of the potential-matches array of $v$ which is related to the class of $u$.

After doing as above for the whole linked list of type-$3$ virtual nodes of layer $j+1$, we now find the maximal matching. For this purpose, for each component of nodes of layers $1$ to $j$, we have one Boolean flag variable which keeps track of whether this component is matched or not in layer $j+1$. We go over the linked list of type-$2$ virtual nodes of layer $j+1$ and for each virtual node $v$ in this list, we go over the edges of $v$ and do as follows (until $v$ gets its class number): for each virtual neighbor $w'$ of $v$, if $w'$ is in a layer in $[1, j]$, we check the component of $w'$. If this component is unmatched and is not deactivated for matching, we check for possibility of matching $v$ to this component. Let $i'$ be the class of this component and let $CID_{w'}$ be the component id of $w'$. Note that we find $CID_{w'}$ using a find operation on the disjoint-set data structure of $w'$. We look in the potential-matches array of $v$ in the entry related to class $i'$. If this linked list has length greater than $1$, or if it has length exactly $1$ and the component id in it is different from $CID_{w'}$, then node $v$ chooses class $i'$. In that case, we also set the matched flag of component of $w'$ to indicate that it is matched now. If $v$ is matched, we are done with $v$ and we go to the next type-$2$ node in the linked list. However, if $v$ does not get matched after checking all of its neighbors, then we choose a random class number for $v$.

It is clear that in the above steps, each edge of the virtual graph that has at least one endpoint in layer $j+1$ is worked on for $O(1)$ times. In each such time, we access $O(1)$ variables and we perform at most one find operation on a disjoint-set data structure. Since each find operation costs $O(1)$ time, the overall cost of these class assignment steps over all the layers becomes $O(m\log^2 n)$.

At the end of the CDS packing construction, in order to turn the CDSs into dominating trees, we simply use a linear time minimum spanning tree algorithm, e.g., \cite{MST}, which on the virtual graph takes $O(m\log^2 n)$ steps.
\end{proof}

\section{Missing Parts of the Dominating Tree Packing Analysis}\label{App:MissingAnalysis}

\begin{proof}[Proof of \Cref{lem:layer1dom}]
  Each virtual node $v \in \mathcal{V}$ has in expectation $\frac{k\log n}{2t} =
  \Omega(\log n)$ virtual neighbors in $\mathcal{V}^i_{\NumOfLayers/2}$. Choosing
  constants properly, the claim follows from a standard Chernoff argument combined with a union bound over all choices of $v$ and over all classes.
\end{proof}

\Cref{fig:ConnectorPaths}, demonstrates an example of potential
connector paths for a component $\mathcal{C}_1 \in
\mathcal{G}[\mathcal{V}^i_\ell]$ (see \Cref{subsec:Connectors}). The figure on the left shows the graph
$G$, where the projection $\Psi(\mathcal{V}^i_\ell)$ is indicated via
green vertices, and the green paths are potential connector paths of
$\Psi(\mathcal{C}_1)$. On the right side, the same potential
connector paths are shown, where the type of the related internal vertices are
determined according to rules (D) and (E) above, and vertices of
different types are distinguished via different shapes (for clarity,
virtual vertices of other types are omitted).

\begin{figure}[t]
\vspace{-0.2cm}
\centering
\includegraphics[width=0.70\textwidth]{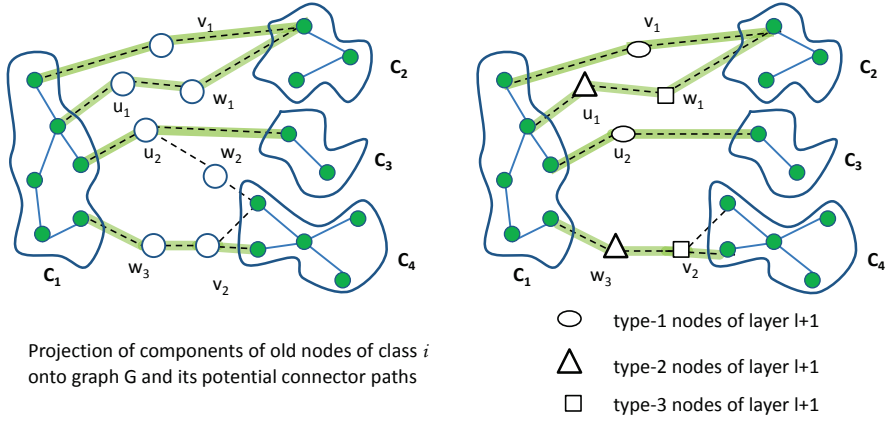}
	\caption{Connector Paths for component $\mathcal{C}_1$ in layer $\ell+1$ copies of $G$}
	\label{fig:ConnectorPaths}
	\vspace{-0.3cm}
\end{figure}

\begin{proof}[Proof of \Cref{lem:CAL}]

Fix a layer $\ell\in [{\NumOfLayers}/{2}, \NumOfLayers-1]$. Fix $\mathcal{V}^i_\ell$ and suppose it is a dominating set of $\mathcal{G}$.

  Consider the projection $\Psi(\mathcal{V}^i_{\ell})$ onto $G$ and
  recall Menger's theorem: Between any pair $(u,v)$ of non-adjacent
  nodes of a $k$-vertex connected graph, there are $k$ internally
  vertex-disjoint paths connecting $u$ and $v$. Applying Menger's
  theorem to a node in $\Psi(\mathcal{C})$ and a node in
  $\Psi(\mathcal{V}^i_\ell\setminus\mathcal{C})$, we obtain at least $k$
  internally vertex-disjoint paths between $\Psi(\mathcal{C})$ and
  $\Psi(\mathcal{V}^i_\ell \setminus \mathcal{C})$ in $G$.
	
	We first show that these paths can be shortened so that they satisfy conditions (B) and (C) of potential connector
  paths, stated in \Cref{subsec:Connectors}.
	
	Pick an arbitrary one of these $k$ paths and denote it $P$ =
  $v_1$, $v_2$, $...$, $v_r$, where $v_1 \in \Psi(\mathcal{C})$ and
  $v_r \in \Psi(\mathcal{V}^i_\ell \setminus \mathcal{C})$. By the
  assumption that $\mathcal{V}^i_\ell$ dominates
  $\mathcal{G}$, since $v_1 \in \Psi(\mathcal{C})$ and $v_r \in
  \Psi(\mathcal{S}^i_\ell \setminus \mathcal{C})$, either there is a
  node $v_i$ along $P$ that is connected to both $\Psi(\mathcal{C})$
  and $\Psi(\mathcal{V}^i_\ell \setminus \mathcal{C})$, or there must
  exist two consecutive nodes $v_i$, $v_{i+1}$ along $P$, such that
  one of them is connected to $\Psi(\mathcal{C})$ and the other is
  connected to $\Psi(\mathcal{V}^i_\ell \setminus \mathcal{C})$. In
  either case, we can derive a new path $P'$ which has at most $2$ internal nodes, i.e., satisfies (B), is internally vertex-disjoint from the other $k-1$ paths since its
  internal nodes are a subset of the internal nodes of $P$ and are not
  in $\Psi(\mathcal{V}^i_{\ell})$.
	
	If in this path with length $2$, the node closer to $\mathcal{C}$ has a neighbor in $\Psi(\mathcal{V}^i_\ell \setminus \mathcal{C})$, or if the node closer to the $\Psi(\mathcal{V}^i_\ell \setminus \mathcal{C})$ side has a neighbor in $\mathcal{C}$, then we can further shorten the path and get a path with only $1$ internal node.	Note that this path would still remain internally vertex-disjoint from the other $k-1$ paths since its internal nodes are a subset of the internal nodes of $P$ and are not
  in $\Psi(\mathcal{V}^i_{\ell})$.
	
	After shortening all the $k$ internally vertex-disjoint paths, we get $k$ internally
  vertex-disjoint paths in graph $G$ that satisfy conditions (A), (B), and (C).

Now using rules (D) and (E) in \Cref{subsec:Connectors}, we can transform these $k$ internally vertex-disjoint \emph{potential} connector paths in $G$ into $k$ internally vertex-disjoint connector paths on the virtual nodes of layer $l+1$. It is clear that during the transition from the real nodes to the virtual nodes, the connector paths remain internally vertex-disjoint.
\end{proof}

\begin{proof}[Proof of \Cref{{lem:maxMatchSize}}]

  In order to prove that the expected size of the maximum matching is
  at least $\Omega(K)$, we focus on a special sub-graph
  $\mathcal{H}$ of the bridging graph (to be described next). We
  show that in expectation, $\mathcal{H}$ has a matching of size
  $\Omega(K)$. Since $\mathcal{H}$ is a sub-graph of the
  bridging graph, this proves that the expected size of the
  maximum matching in the bridging graph is at least
  $\Omega(K)$ and thus completes the proof of
  \Cref{lem:FML3}.
\smallskip

\paragraph{Subgraph $\mathcal{H}$} This graph is obtained from the
connector paths of components in $\mathcal{K}$. First, \emph{discard} each new
node of type $3$ with probability $1/2$. This is done for cleaner
dependency arguments. For each type-$2$ new node $v$, we determine
components that are neighbors of $v$ in $\mathcal{H}$ as follows:
Consider all the long connector paths for components in $\mathcal{K}$
that go through $v$. Pick an arbitrary one of these long connector
paths and assume that it belongs to component $\mathcal{C}\in
\mathcal{K}$ of $\mathcal{G}[\mathcal{V}^i_{\ell}]$. Suppose that the path goes from $\mathcal{C}$ to $v$, then to
a type-$3$ node $u$, and then finally to a component $\mathcal{C}'
\neq \mathcal{C}$ of the graph $\mathcal{G}[\mathcal{V}^i_{\ell}]$. Mark component
$\mathcal{C}$ as a potential neighbor for $v$ in $\mathcal{H}$ if and only if
$w$ is not discarded and $w$ has joined class $i$. Go over all long
connector paths of $v$ and mark the related potential component
neighbors of $v$ accordingly. If at the end, $v$ has exactly one
potential component neighbor, then we include that one as the neighbor
of $v$ in $\mathcal{H}$. Otherwise, $v$ does not have any neighbor in
$\mathcal{H}$.

It is easy to see see that $\mathcal{H}$ is a sub-graph of the bridging
graph. Moreover, the degree of each type-$2$ new node in $\mathcal{H}$
is at most $1$. However, we remark that it is possible that a
component has degree greater than one in $\mathcal{H}$. Thus,
$\mathcal{H}$ is not necessarily a matching.

  To complete the proof, in the following, we show that the expected size of the maximum matching of $\mathcal{H}$ is at least $\Omega(K)$.

  More specifically, we show that there is a constant $\delta>0$ such that for each
  component $\mathcal{C}$ in $\mathcal{K}$, with probability
  at least $\delta$, $\mathcal{C}$ has at least one long connector
  path $p$ that satisfies the following condition:
  \begin{itemize}[noitemsep,topsep=8pt,parsep=5pt,partopsep=0pt]
  \item[\textbf{(\boldmath$\star$)}] The long connector path $p$ has internal
    type-$2$ new node $v$, and in $\mathcal{H}$, node $v$ has
    $\mathcal{C}$ as its only neighbor.
  \end{itemize}
  If a component $\mathcal{C}$ has at least one long connector path
  that satisfies ($\star$), then we pick exactly one such long connector
  path and we match the type-$2$ node of that path to
  $\mathcal{C}$.
	
	Once we show that each component $\mathcal{C}$ in $\mathcal{K}$ with probability at least $\delta$ has a long connector path satisfying ($\star$), then the proof can be completed by
  linearity of expectation since the number of components in $\mathcal{K}$ is $K$ and each components in $\mathcal{K}$ gets matched with probability at least $\delta$.

  We first study each long connector path $p$ of $\mathcal{C}$
  separately and show that $p$ satisfies ($\star$) with probability at least $\frac{1}{4t}$. Moreover, we show
  that regardless of what happens for other long connector paths of
  $\mathcal{C}$, the probability that $p$ satisfies ($\star$) is at
  most $\frac{1}{2t}$.

  Suppose that $p$ is composed of type-$2$ new node $v$ and type-$3$
  new node $w$. Suppose that other than class $i$, $v$ is also on long
  connector paths of classes $i'_1$, $i'_2$, $\dots$, $i'_z$ where $z
  < t$. By \Cref{prop:longpaths}, for each other
  class $i'_j$, $v$ is on a connector path of at most one component of
  class $i'_j$. Let $u_1$ to $u_{z'}$ be the type-$3$ nodes on the
  long connector paths related to these classes. Note that $z'$ might
  be smaller than $z$ as it is possible that the long connector paths
  of the $z$ classes share some of the type $3$ nodes. Path $p$
  satisfies condition ($\star$) if and only if the following two conditions hold:
  (a) $w$ is not discarded and it joins class $i$, (b) for each class
  $i'_{j}$, the type-$3$ node on the long connector path related to
  class $i'_{j}$ that goes through $v$ is either discarded or it does
  not join class $i'_{j}$. The probability that (a) is satisfied is
  exactly $\frac{1}{2t}$. On the other hand, since different classes
  might have common type-$3$ nodes on their paths, the events of
  different classes $i'_{j}$ satisfying the condition (b)
  are not independent. However, for each type-$3$ new node $u_{j'}$,
  suppose that $x_{j'}$ is the number of classes other than $i$ which
  have long connector paths through $u_{j'}$. Then, the probability
  that $u_{j'}$ is discarded or that it does not join any of these
  $x_{j'}$ classes is $1-\frac{x_{j'}}{2t} \geq
  4^{-\frac{x_{j'}}{2t}}$, where the inequality follows because
  $\frac{x_{j'}}{2t} \leq \frac{1}{2}$. The probability that the above
  condition is satisfied for all choices of $u_{j'}$ is at least
  $4^{-\sum_{j'=1}^{z} \frac{x_{j'}}{2t}}$. Since $\sum_{j' = 1}^{z'}
  {x_{j'}} = z \leq t$, we get that the probability that (b) holds is
  at least $4^{-\,\frac{1}{2}} = \frac{1}{2}$.
  Hence, the probability that both (a) and (b) happen is at least
  $\frac{1}{4t}$. This proves that $p$ satisfies ($\star$) with probability at least
  $\frac{1}{4t}$. To show that this probability is at most
  $\frac{1}{2t}$, regardless of what happens with other paths, it is
  sufficient to notice that $w$ satisfies (a) with probability at most $\frac{1}{2t}$.

  We now look over all long connector paths of component
  $\mathcal{C}$ together. Let $Z$ be the number of long connector paths of $\mathcal{C}$ which satisfy
  condition ($\star$). To conclude the proof, we need to show that $\Pr[Z\geq 1]\geq\delta$ for some constant $\delta>0$.
	Note that the events of satisfying this condition for different long connector paths are not independent. In
  fact, they are positively correlated and thus we can not use standard
  concentration bounds like a Chernoff bound. Markov's inequality
  does not give a sufficiently strong result either. To prove the
  claim, we use an approach which has a spirit similar to the proof of
  Markov's inequality but is tailored to this particular case.
	
  We know that w.h.p., $\mathcal{C}$ has at least $k' =
  \Omega(k)$ long connector paths. Let us assume that
  this holds. Using linearity of expectation, we have $\E[Z]
  \geq \frac{k'}{4t}$. Since $t = \Theta(k)$, by choosing a small enough constant in definition of $t=\Theta(k)$, we get that $\E[Z] \geq \frac{k'}{4t}=z_0$ for some constant
  $z_0>1$. Given this, we want to show that $\Pr[Z\geq 1]\geq\delta$ for some constant $\delta>0$.
	
	Because of the upper bound on the probability for a path
  $p$ to satisfy condition ($\star$) which holds independently of what
  happens for other connector paths, we have $$\Pr [Z=\zeta] \leq
  \binom{k'}{\zeta} (\frac{1}{2t})^{\zeta} \leq \left(\frac{2ek'}{2t\zeta}\right)^{\zeta} = \left(\frac{4ez_0}{\zeta}\right)^{\zeta}.$$
	Following the above equation, intuitively, for some constant threshold $\zeta_0$ and a variable $\zeta \geq \zeta_0$, $\Pr[Z=\zeta]$
  decreases exponentially. This happens for example if we set $\zeta_0=20 z_0$. Hence, the contribution of the part where
  $Z>\zeta_0$ to the expectation $\E[Z]$ is very small and essentially negligible. This means
  that to have $\E[Z] \geq z_0$, a constant part of the probability mass should be on values $Z \in [1, \zeta_0]$, which completes the proof.
	
	Having this intuition, the formal argument is as follows. Let
  $\beta=\Pr[Z\geq 1]$. Then we have
  \begin{eqnarray}
    z_0 \leq \E[Z] & = & \sum_{\zeta=0}^{\infty} \zeta \Pr[Z=\zeta]  = \sum_{\zeta=1}^{\zeta_0} \zeta \Pr[Z=\zeta] +  \sum_{\zeta=\zeta_0+1}^{\infty} \zeta \Pr[Z=\zeta]  \nonumber\\
    & \leq & \beta \cdot \zeta_0 +
    \sum_{\zeta=\zeta_0+1}^{\infty} \zeta \binom{k'}{\zeta} \left(\frac{1}{2t}\right)^{\zeta} \leq \beta \cdot \zeta_0 +
    \sum_{\zeta=\zeta_0+1}^{\infty} \zeta \left(\frac{2ek'}{2t\zeta}\right)^{\zeta} \nonumber \\
		&\leq & \beta \cdot \zeta_0 +  \sum_{\zeta=\zeta_0+1}^{\infty} \zeta  \left(\frac{4ez_0}{\zeta}\right)^{\zeta}< \zeta_0 (\beta + \frac{1}{2^{10}}), \nonumber \label{eq:Zbound}
  \end{eqnarray}
  where the last inequality holds if constant $\zeta_0$ is chosen sufficiently
  large---e.g. $\zeta_0=20 z_0$. We get that $\beta=\Pr[Z\geq 1] \geq \frac{1}{20}-\frac{1}{2^{10}}$. This shows that  $\Pr[Z\geq 1]\geq\delta$ for some constant $\delta>0$ and thus completes the proof.

\end{proof}

\begin{proof}[proof of \Cref{lem:size}]
  Consider an arbitrary virtual node $v$. Either (a) $v$
  chooses its class number randomly, or (b) $v$ is a type-$2$ node and
  it chooses its class number based on the maximal matching. Using a
  Chernoff bound, the total number of virtual nodes that join class $i$
  randomly---following condition (a)---is $O(\frac{n\log n}{t}) = O(\frac{n\log n}{k})$
  w.h.p. The number of virtual nodes that join class $i$ following
  condition (b) is at most equal to the number of connected components of
  $G^i_{\NumOfLayers/2}$. Since virtual nodes of layers $1$ to $\NumOfLayers/2$ choose their classes following condition (a),
  we get that the number of virtual nodes that join class $i$ following condition (b) is also $O(\frac{n\log n}{k})$ w.h.p.
\end{proof} 
\section{Testing A Dominating Tree Packing}\label{sec:App6}

\begin{lemma}
A dominating tree packing of a connected undirected graph $\mathcal{G}=(\mathcal{V}, \mathcal{E})$ can be tested, using a distributed algorithm in $\tilde{O}(\min\{d', diam(\mathcal{G}) + \sqrt{|\mathcal{V}|}\})$ rounds of the \LBM model, where $d'$ is an upper bound on the diameter of each dominating tree, or using a centralized algorithm in $\tilde{O}(\mathcal{E})$ steps,
\end{lemma}
More specifically, the lemma states the following: Suppose that we are given a partition of vertices $\mathcal{V}$ into disjoint classes $V_1$, $V_2$, \dots, $V_t$ where
each node knows its class number and the value of $t$. We can simultaneously test whether it is true for all classes $i\in [1,t]$ that $\mathcal{G}[V_i]$ is a CDS, or not. If each class is a CDS, then the test passes and otherwise---i.e., if there is even one class that is not a CDS---then the test fails with high probability. Moreover, the outputs of all nodes are consistent in that either all the nodes declare a failure or the test passes in all the nodes.

\begin{proof}
We first explain the distributed algorithm. The centralized algorithm is a simpler variant of the same approach. The general idea is to first check connectivity of all classes, and then check whether there is any disconnected class or not. Let $D=diam(\mathcal{G})$, $n'=|\mathcal{V}|$, and $m'=|\mathcal{E}|$.

  \paragraph{Distributed Domination Test} We first check if each class is a dominating set. For this, each node sends its
  class number to its neighbors. If a node $v$ is not dominated by a
  class $i$, that is if $v$ does not receive any message from a node in class $i$, then $v$ initiates a \emph{`domination-failure'} message
  and sends it to its neighbors. We use $\Theta(D)$ rounds to propagate these \emph{`domination-failure'} messages: in each round,
  each node sends the \emph{`domination-failure'} message to its
  neighbors if it received \emph{`domination-failure'} message in one
  of the previous rounds. After these $\Theta(D)$ rounds, if the domination part of the test passes, we check for connectivity.
	
	\paragraph{Distributed Connectivity Test} We first use $O(\min\{d', Diam(\mathcal{G}) +\sqrt{n'}\log^*n\})$ rounds
  to identify the connected components of each class, using \Cref{thm:ComponentID} where each node $v$ starts with its own id as its variable $x_v$ and only edges between the nodes of the same class are included in the subgraph $\mathcal{G}_{sub}$. Hence, at the end of this part, each connected component has a leader and every node $u$ knows the id of the leader of its connected component, which is recorded as the component id of $u$.
	
	Given these component ids, to test connectivity, we check
  if there exist two nodes in the same class with different component
  ids. Suppose that there exists a nonempty set of classes $I$ which
  each have two or more connected components. We show a protocol
  such that w.h.p., at least one node $v$ receives two different component ids related to a class in $i \in
  I$. We call this a \emph{``disconnect detection"} as it indicates that class
  $i$ is disconnected. If this happens, then $v$ initiates a
  \emph{`connectivity-failure'} message. $\Theta(D)$
  rounds are used to propagate these \emph{`connectivity-failure'}
  messages.

  In the first round of the connectivity test, each node sends its
  class number and its component id to all of its neighbors. Since
  each class is dominating (already tested), each node receives at
  least one component id for each class. If a disconnect is
  detected at this point, we are done. Suppose that this is not the
  case. Note that this is possible because connected components of
  each class $i \in I$ can be at distance more than $1$ from each
  other.

  However, using Menger's theorem along with vertex connectivity $k$
  of the graph and since the domination part of the test has passed,
  with an argument as in the proof of \Cref{lem:CAL}, we get that for each class $i \in
  I$ and each component $\mathcal{C}$ of class $i$, there are $k$ internally vertex-disjoint
  paths of length exactly $3$ connecting $\mathcal{C}$ with other
  components of class $i$. Note that the length is exactly $3$ because
  length-$2$ would lead to detection of inconsistency in the first
  part of the connectivity test. Let us call these \emph{detector
    paths} of class $i$.
		
	The algorithm is as follows: in each round, each node $v$ chooses a random class $i'$ and sends the component ID
  related to class $i'$ (the component ID related to class $i'$ that
  $v$ has heard so far). In order for the inconsistency to be
  detected, it is enough that one of the internal nodes on the (at least)
  $k$ detector paths related to a class $i\in I$ sends the component
  ID of class $i$ that it knows. This is because, if that happens, then the other internal node on that path would
  detect the disconnect.
	
	For each node $v$, let $x_v$ be the number of disconnected classes $i$ for which $v$ is an internal node on one
  of the detector paths of class $i$. Then, in each round, with
  probability $\frac{x_v}{t}$, node $v$ sends a component ID which
  leads to disconnect detection. Hence, for each round, the
  probability that no such ID is sent is
  \begin{eqnarray}
    \prod_{v\in V}  \left(1-\frac{x_v}{t}\right) \leq& e^{-\sum_{v\in V}  \frac{x_v}{t}} \stackrel{(\dagger)}{\leq} e^{-\frac{2k \cdot |I|}{t}} \stackrel{(*)}{\leq}& e^{-\frac{2k \,\cdot\, \max\{1, t-k\}}{t}} < e^{-\frac{1}{2}}. \nonumber
  \end{eqnarray}
  Here, Inequality ($\dagger$) holds because there are $|I|$
  disconnected classes and each disconnected class has at least $2k$
  internal nodes on its detector paths. Inequality ($*$) holds
  because a graph with vertex connectivity $k$ can have at most
  $k$ vertex-disjoint CDS sets and thus $|I| \geq t-k$, and we
  have assumed that $I \neq \emptyset$.
  Since in each round there is a constant probability for disconnect
detection, after $\Theta(\log n')$ rounds, at least one node will detect it with high probability, and thus
  after additional $\Theta(D)$ rounds, all nodes know that at least one
  class is not connected. If no such disconnect is detected in
  initial $\Theta(\log n')$ rounds (thus not reported by the end of
  $\Theta(D + \log n')$ rounds), the connectivity test also passes and thus, the complete CDS partition test passes claiming that w.h.p., each class is a CDS.

\paragraph{The Centralized Tests} Now we turn to explaining the centralized counterpart of the above algorithm: Testing domination in $O(m')$ time is easy: we go over the nodes one by one, for each node, we read the class number of its neighbors and record which classes are dominating this node. After that, if there is any class left out, we have found `domination-failure'. This way, we work on each edge at most twice, once from each side, and thus the whole domination testing finishes in $O(m')$ steps. 	For testing connectivity, the general approach remains the same as in the distributed setting, but we change the component identification part. Note that in the centralized setting, one can identify the connected components of a subgraph of the graph $\mathcal{G}$ in $O(m')$ rounds, using disjoint-set data structures (see \Cref{sec:centPack}). After identifying the components, we can deliver the component id of each node to its neighbors in a total of $O(m')$ rounds. Then, we simply run the $\Theta(\log n')$ rounds of the distributed algorithm where each node sends the id of a random class to its neighbors in a centralized manner. Each round can be clearly simulated in $O(m')$ steps of the centralized setting. Hence, $O(\log n')$ rounds can be simulated in $O(m'\log n')$ rounds and after that, if there is any disconnected class, with high probability a disconnect detection has happened. This concludes the centralized test.
\end{proof}

\section{Missing Parts of the Fractional Spanning Tree Packing}\label{app:STpack}
We first present the analysis for the case $\lambda=O(\log n)$, for which we presented the algorithm in \Cref{subsec:STsmall}. Then, we present the proof of \Cref{lem:STpackingMST}.

\paragraph{Analysis for the Algorithm of \Cref{subsec:STsmall}} First, in \Cref{lem:TerminationQuality} we show that if in some iteration we stop because of the condition $\sum_{e\in MST} c_e > (1-\eps) \sum_{e\in E} c_e \cdot x_e$, then  $\max_{e\in E} z_e \leq 1+\eps$. Then, in \Cref{lem:terminationSpeed}, we show that if throughout $\Theta(\log^3 n)$ iterations, the condition $\sum_{e\in MST} c_e > (1-\eps) \sum_{e\in E} c_e \cdot x_e$ is never satisfied, then the collection attained at the end of $\Theta(\log^3 n)$ iterations has the property that $\max_{e\in E} z_e \leq 1+\eps$.

\begin{lemma}\label{lem:TerminationQuality}  If in some iteration $\sum_{e\in MST} c_e > (1-\eps) \sum_{e\in E} c_e \cdot x_e$, then $\max_{e\in E} z_e \leq 1+6\eps$.
\end{lemma}

\begin{proof} Let $Z= \max_{e\in E} z_e$. First note that
\begin{eqnarray}\sum_{e\in E  \ and \ z_e \leq (1-\eps) Z} c_e \leq &&\sum_{e\in E} exp(\alpha(1-\eps)Z) \nonumber \\
\leq &&m \cdot exp(-\alpha\eps Z) \cdot exp(\alpha Z) \leq m \cdot exp(-\alpha\eps Z) \sum_{e\in E} c_e \leq (\eps/2) \cdot \sum_{e\in E} c_e. \nonumber
\end{eqnarray}
Thus, we have
\begin{eqnarray}
\sum_{e\in E} c_e \cdot x_e \geq \sum_{e\in E \ and \ z_e \geq (1-\eps) Z} c_e \cdot x_e \geq (1-\eps)\frac{Z}{\ceil{\frac{\lambda-1}{2}}} \sum_{e\in E  \ and \ z_e \geq (1-\eps) Z} c_e  \geq (1-\eps)^2  \frac{Z}{\ceil{\frac{\lambda-1}{2}}} \sum_{e\in E} c_e, \nonumber
\end{eqnarray}
and hence
\begin{eqnarray}
\centering
\sum_{e\in MST} c_e > (1-\eps) \frac{1}{\ceil{\frac{\lambda-1}{2}}} \sum_{e\in E} c_e \cdot x_e > (1-\eps)^3  \frac{Z}{\ceil{\frac{\lambda-1}{2}}} \sum_{e\in E} c_e\geq  (1-\eps)^3 Z \sum_{e\in MST} c_e, \nonumber
\end{eqnarray}
where the last inequality follows from the results of Tutte and NashWillams, which show that $E$ contains at least $\ceil{\frac{\lambda-1}{2}}$ edge-disjoint spanning trees and clearly each of these trees has cost at least equal to that of the MST. Comparing the two sides of the above inequality, we get $Z \leq (1-\eps)^{-3} \leq 1+6\eps$.
\end{proof}

\begin{lemma}\label{lem:terminationSpeed}If the condition $\sum_{e\in MST} c_e > (1-\eps) \sum_{e\in E} c_e \cdot x_e$ is never satisfied in $\Theta(\log^3 n)$ iterations of the algorithm, then for the collection attained at the end of $\Theta(\log^3 n)$ iterations, we have $\max_{e\in E} z_e \leq 1+\eps$.
\end{lemma}
\begin{proof}
Consider the potential function $\Phi = \sum_{e\in E} c_e= \sum_{e\in E} exp(\alpha z_e)$. We first show that, if in an iteration we have $\sum_{e\in MST} c_e \leq (1-\eps) \sum_{e\in E} c_e \cdot x_e$, then with the update of this iteration, the potential function decreases at least by a factor of $1-\Theta(\eps/\log n)$.
\begin{eqnarray}\Delta\Phi =&&\Phi_{old} - \Phi_{new} \nonumber \\
=&&\sum_{e\in E} exp(\alpha z^{old}_e)-exp(\alpha z^{new}_e) = \sum_{e\in E} exp(\alpha z^{old}_e) \cdot (1- exp(\alpha \beta \ceil{\frac{\lambda-1}{2}}\cdot (1^{MST}_e - x^{old}_e)))\nonumber \\
\geq &&\alpha \beta \ceil{\frac{\lambda-1}{2}} \sum_{e\in E} exp(\alpha z^{old}_e) \cdot (x^{old}_e -1^{MST}_e) = \alpha\beta \ceil{\frac{\lambda-1}{2}} (\sum_{e\in E} c_e \cdot x_e -\sum_{e\in MST} c_e ) \nonumber\\
\geq &&\alpha \beta \ceil{\frac{\lambda-1}{2}} \eps \sum_{e\in E} c_e \cdot x_e \geq \alpha \beta \ceil{\frac{\lambda-1}{2}}\eps \sum_{e\in E \ and \ z_e \geq (1-\eps) Z} c_e \cdot x_e \nonumber \\
\geq &&\alpha \beta \eps  (1-\eps) \cdot Z\sum_{e\in E  \ and \ z_e \geq (1-\eps) Z} c_e \geq \alpha \beta \eps (1-\eps)^2  Z \sum_{e\in E} c_e \geq \Theta(\frac{\eps}{\log n})\Phi_{old}.\nonumber
\end{eqnarray}

Now note that the starting potential is at most $m \cdot exp(\alpha \ceil{\frac{\lambda-1}{2}})$. When the potential falls below $exp(\alpha(1+\eps))$, all edges have $z_e \leq 1+\eps$ which means we have found the desired packing. Since in each iteration that condition $\sum_{e\in MST} c_e \leq (1-\eps) \sum_{e\in E} c_e \cdot x_e$ holds, the potential decreases by a factor of $1-\Theta(\eps/\log n)$, we get that after at most $\Theta(\frac{\log n}{\eps} \cdot (\alpha\lambda+\log m))$ iterations, it falls below $exp(\alpha(1+\eps))$. Noting that $\alpha = O(\log n)$, $\lambda=O(\log n)$ and $\eps=\Theta(1)$, we can infer that this happens after at most $\Theta(\log^3 n)$ iterations.
\end{proof}

\bigskip

\begin{proof}[Proof of \Cref{lem:STpackingMST}]
We need to first briefly review the the general approach of ~\cite{KuttenPeleg95}. The algorithm of~\cite{KuttenPeleg95} first uses $O(d\log^*)$ rounds to get a $d$-dominating set $T$ with size at most $O(\frac{n}{d})$ and a partition of the graph into clusters of radius at most $d$ around each node of $T$, where also each of these clusters is spanned by a fragment of the minimum spanning tree. Thus, the part of the minimum spanning tree that is completely inside one fragment is already determined. It then remains to determine the MST edges between different fragments. This part is performed by a pipe-lined upcast of the inter-fragment edges on a breadth first search and it is shown that this upcast takes at most $O(D+ \frac{n}{d})$ rounds, where $O(\frac{n}{d})$ is the number of the inter-fragment edges in the MST. At the end, $O(\frac{n}{d})$ inter-fragment edges are broadcast to all nodes. Choosing $d=\sqrt{n}$ then leads to time complexity of~\cite{KuttenPeleg95}.

In our problem, we solve $\eta$ MSTs of subgraphs $H_1$ to $H_\eta$ in parallel. The first part of creating the local fragments of MST is done in each subgraph independently, as they are edge-disjoint, in $O(d\log^*n)$ rounds. However, we must not do the upcasts on the BFS trees of subgraphs $H_1$ to $H_\eta$ as each of these subgraphs might have a large diameter. Instead, we perform all the upcasts on the same BFS tree of the whole graph. It is easy to see that we can pipe-line the inter-fragment edges of different MSTs so that they all arrive at the root of this BFS after at most $O(D+ \eta\frac{n}{d})$ rounds. Choosing $d=\sqrt{n\eta}$ gives us that we can simultaneously run one iteration of the fractional spanning tree packing of each subgraph, all together in time $O(D+ \sqrt{n\eta}\log^*n)$. Since we have at most $\Theta(\log^3 n)$ iterations in the fractional spanning tree packing, and as $\eta=\Theta(\frac{\lambda}{\log n})$, the total round complexity becomes at most $O((D+ \sqrt{\frac{n\lambda}{\log n}}\log^*n)\log^{3} n)$.
\end{proof}

\section{Lower Bounds}
\label{sec:lower}

In this section, we present the distributed lower bounds on finding
fractional dominating tree packings or fractional spanning tree
packings with size approximately equal to connectivity. Formally, we give lower bounds on approximating the value of
the vertex or edge connectivity of a graph. The lower bounds about tree packings are then
obtained because given a (fractional) dominating tree or spanning tree
packing of a certain size---which is promised to be an approximation of connectivity---all nodes can immediately obtain an approximation of vertex or edge connectivity.

The lower bound for approximating the edge connectivity of a graph
in the \CM\ model already appears in \cite{GK13_arxiv}. We just
restate it here together with the implication on computing
(fractional) spanning tree packings. 

\begin{theorem}\label{thm:edgelowerbound}[Theorem 6.4 of \cite{GK13_arxiv}]
  For any $\alpha>1$ and $\lambda\geq 1$, even for diameter
  $D=O\big(\frac{1}{\lambda\log
    n}\cdot\sqrt{\frac{n}{\alpha\lambda}}\big)$, distinguishing
  networks with edge connectivity at most $\lambda$ from networks with
  edge connectivity at least $\alpha\lambda$ requires at least
  $\Omega\big(D+\frac{1}{\log n}\sqrt{\frac{n}{\alpha\lambda}}\big)$
  rounds in the \CM\ model. The same lower bound applies to computing (fractional)
  spanning tree packings of size larger than $n/(\alpha\lambda)$,
  where $\lambda$ is the edge connectivity of the network.
\end{theorem}

For vertex connectivity, we even get the following stronger lower
bound.
\begin{theorem}\label{thm:vertexlowerbound}
  For any $\alpha>1$ and $k\geq 4$, even in networks of diameter $3$,
  distinguishing networks with vertex connectivity at most $k$ from
  networks with vertex connectivity at least $\alpha k$ requires at
  least $\Omega\big(\sqrt{n/(\alpha k \log n)}\big)$ rounds in the
  \LBM\ model. The same lower bound also applies to computing
  (fraction) dominating tree packings of size larger than
  $n/(k\alpha)$ or for finding a vertex cut of size at most
  $\min\set{\delta\cdot\sqrt{n/(\alpha k \log n)},\alpha\cdot k}$, for
some constant $\delta>0$ and where $k$ is the vertex connectivity of
the network.
\end{theorem}

In the remainder of the section, we prove Theorem
\ref{thm:vertexlowerbound}.  Both lower bounds (Theorems
\ref{thm:edgelowerbound} and \ref{thm:vertexlowerbound}) are based on
the approach used in \cite{dassarma12}. However, since all the lower
bounds in \cite{dassarma12} are for the \CM\ model, in order to get
the slightly stronger bound of Theorem \ref{thm:vertexlowerbound}, we
need to adapt to the node capacitated \LBM\ model.

The lower bound is proven by a reduction from the $2$-party set
disjointness problem. Assume that two players Alice and Bob get two
sets $X$ and $Y$ as inputs. If the elements of sets are from a
universe of size $N$, it is well known that determining whether $X$
and $Y$ are disjoint requires Alice and Bob to exchange $\Omega(N)$
bits \cite{kalyanasundaram92,razborov92}. This lower bound even holds
if Alice and Bob are promised that $|X\cap Y|\leq 1$
\cite{razborov92}, it even holds for randomized protocols with
constant error probability and also if Alice and Bob only have access
to public randomness (i.e., to a common random source). Note that this
immediately also implies an $\Omega(N)$ lower bound on the problem of
finding $X\cap Y$, even if Alice and Bob know that $X$ and $Y$
intersect in exactly one element. In fact, if Alice and Bob even need
to exchange $\Omega(N)$ bits in order to solve the following
problem. Alice is given a set $X$ as her input and Bob is given a set
$Y$ as his input, with the promise that $|X\cap Y|=1$. Alice needs to
output a set $X'\subseteq X$ and Bob needs to output a set
$Y'\subseteq Y$ such that $X\cap Y\subseteq X'\cup Y'$ and such that
$|X'\cup Y'|\leq cN/\log_2 n$ for an appropriate constant $c>0$. Given
such sets $X'$ and $Y'$, Alice can just send $X'$ to Bob using
$|X'|\cdot \log_2N\leq cN$ bits. For a sufficiently small constant
$c>0$, that is at most a constant fraction of the bits that are needed
to find $X\cap Y$.

\subsection{Lower Bound Construction}
\label{sec:lowerconstruction}

We next describe the construction of a family $\mathcal{G}$ of
networks that we use for our reductions from the above variants of the
set disjointness problem. Instead of directly defining $\mathcal{G}$,
it is slightly easier to first introduce a construction $\mathcal{H}$
for weighted graphs. Eventually, nodes of weight $w\geq 1$ will be
replaced by cliques of size $2$ and edges are replaced by complete
bipartite subgraphs. The weighted graph family $\mathcal{H}$ is based
on two integer parameters $h\geq 2$ and $\ell\geq 1$ and a positive
(integer) weight $w>1$.  The family contains a graph $H(X,Y)\in
\mathcal{H}$ for every set $X\subseteq[h]$ and for every
$Y\subseteq[h]$ (i.e., for every possible set disjointness input for
sets over the universe $[h]$). The node set $V_H(X,Y)$ of $H(X,Y)$ is
defined as
\[
V_H(X,Y) := \set{0,\dots,h}\times[2\ell]\cup\set{a,b}\cup
V_X\cup V_Y,
\]
where $V_X:=\set{u_x: x\in X}$ and $V_Y:=\set{v_y: y\in Y}$. Hence,
$V_H(X,Y)$ contains a node $(p,q)$ for every $q\in\set{0,\dots,h}$ and
every $p\in[2\ell]$, a node $u_x$ for each $x\in X$, a node $v_y$ for
each $y\in Y$, and two additional nodes $a$ and $b$. All the nodes
$(p,q)$ (for $(p,q)\in\set{0,\dots,h}\times[2\ell]$) have weight $w$,
all other nodes have weight $1$. The edges of $H(X,Y)$ are defined as
follows. First, the ``heavy'' nodes $(p,q)$ are connected to $h+1$
disjoint paths by adding an edge between $(p,q)$ and $(p,q+1)$ for
each $p\in\set{0,\dots,h}$ and each $q\in\set{1,\dots,2\ell-1}$. The
nodes $u_x$ and $v_y$ are used to encode a set disjointness instance
$(X,Y)$ into the graph $H(X,Y)$. For every $x\in X$, node $u_x$ is
connected to node $(0,1)$ (the first node of path $0$) and to node
$(x,1)$ (the first node of path $x$). In addition, for all $x'\not\in
X$, node $(0,1)$ is directly connected to node $(x',1)$ (the first
node of path $x'$). We proceed similarly with the nodes $v_y\in
V_Y$. For every $y\in Y$, node $u_y$ is connected to node $(0,2\ell)$
(the last node of path $0$) and to node $(y,2\ell)$ (the last node of
path $y$). In addition, for all $y'\not\in Y$, node $(0,2\ell)$ is
directly connected to node $(y',1)$ (the last node of path
$y'$). Finally, we use the nodes $a$ and $b$ in order to get a graph
with small diameter. The two nodes are connected by an edge and every
other node of the graph is either connected to $a$ or to
$b$. Basically, the left half of the graph is connected to node $a$
and the right half of the graph is connected to $b$. Formally, all
nodes $u_x\in V_X$ and all nodes $(p,q)$ for all $q\leq \ell$ are
connected to node $a$. Symmetrically, all nodes $v_y\in V_Y$ and all
nodes $(p,q)$ for $q>\ell$ are connected to node $b$.  An illustration
of $H(X,Y)$ is given in Figure \ref{fig:lowerboundgraph}.

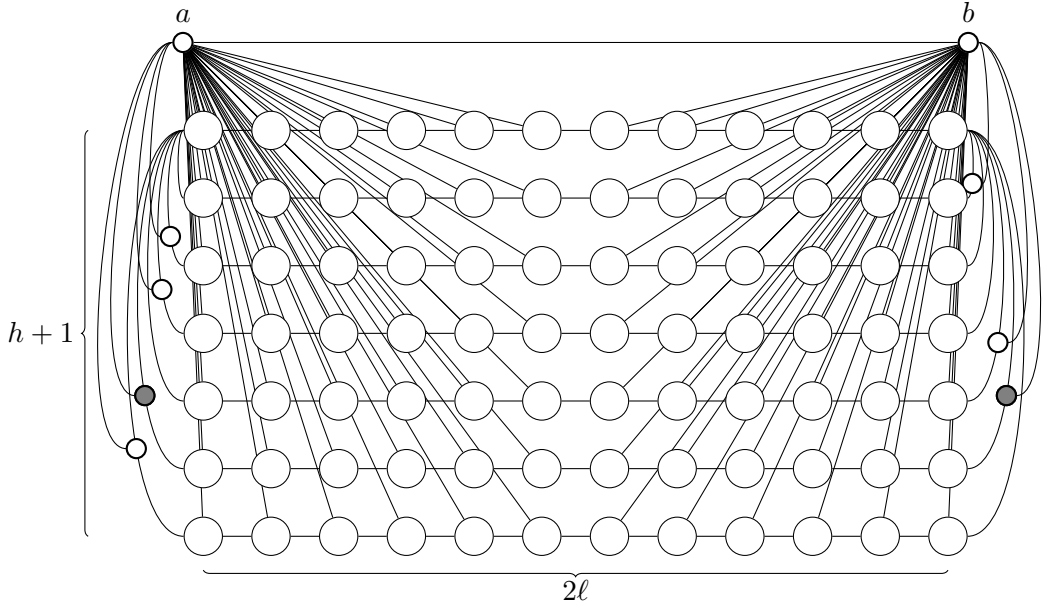
\begin{figure}[t]
  \centering

  \def\pseed{113}
\def\scaling{0.9}
\def\npaths{7}
\def\pathlength{12} 
\def\bignoderadius{\scaling*2mm}
\def\smallnoderadius{\scaling*1mm}
\def\curveradiusA{0.2cm}
\def\curveradiusB{0.15cm}
\newarray\randomlside
\newarray\randomrside
  
\tikzstyle{bignode}=[shape=circle, inner sep=\bignoderadius, fill=white]
\tikzstyle{smallnode}=[shape=circle, thick, inner sep=\smallnoderadius, fill=white]

\begin{tikzpicture}[scale=\scaling]
  \pgfmathsetseed{\pseed}

  \node[smallnode] (lroof) at (0.7,0.3) {};
  \node[smallnode] (rroof) at (\pathlength+0.3,0.3) {};
  
  \foreach \i in {1,...,\npaths}
  {
    \foreach \j in {1,...,\pathlength}
    {
      \node[bignode] (bignode_\i_\j) at (\j,-\i) {};

      \pgfmathtruncatemacro{\middlecolumn}{\pathlength/2+1.0}
      \ifthenelse{\j < \middlecolumn}
      {
        \draw[-,thin] (lroof) -- (bignode_\i_\j);
      }{
        \draw[-,thin] (rroof) -- (bignode_\i_\j);
      }

      \ifthenelse{\j > 1}{
        \pgfmathtruncatemacro{\prevj}{\j-1}
        \draw[-] (bignode_\i_\prevj) -- (bignode_\i_\j);
      }
      {}
    }
  }

  \node[bignode] (lcorner) at (bignode_1_1) {};
  \node[bignode] (rcorner) at (bignode_1_\pathlength) {};
  
  \foreach \i in {2,...,\npaths}
  {
    \draw[-] (lcorner) .. controls +(left: \i*\curveradiusA) and +(left: \i*\curveradiusA) .. (bignode_\i_1) node[pos=0.7] (lcurvenode_\i) {};
    \draw[-] (rcorner) .. controls +(right:\i*\curveradiusA) and +(right:\i*\curveradiusA) .. (bignode_\i_\pathlength) node[pos=0.7] (rcurvenode_\i) {};
  }      

  \foreach \i in {2,...,\npaths}
  {
    \randomlside(\i)={0}
    \randomrside(\i)={0}
    \pgfmathtruncatemacro{\LorR}{rnd<0.5}
    \ifthenelse{\LorR = 1}
    {
      \randomlside(\i)={1}
    }{
      \randomrside(\i)={1}
    }
  }
  \pgfmathtruncatemacro{\specpath}{rnd*(\npaths-1)+2}
  \randomlside(\specpath)={1}
  \randomrside(\specpath)={1}
  
  \draw[-] (lroof) -- (rroof);
  
  \foreach \i in {2,...,\npaths}
  {
    \checkrandomlside(\i)\let\vls\cachedata
    \checkrandomrside(\i)\let\vrs\cachedata
    \ifthenelse{\vls = 1}
    {
      \draw[-] (lroof) .. controls +(left: \i*\curveradiusB) and +(left: \i*\curveradiusB) .. (lcurvenode_\i);
    }{}
    \ifthenelse{\vrs = 1}
    {
      \draw[-] (rroof) .. controls +(right: \i*\curveradiusB) and +(right: \i*\curveradiusB) .. (rcurvenode_\i);
    }{}
  }

  \foreach \i in {2,...,\npaths}
  {
    \checkrandomlside(\i)\let\vls\cachedata
    \checkrandomrside(\i)\let\vrs\cachedata
    \ifthenelse{\vls = 1}
    {
      \node[smallnode,draw] at (lcurvenode_\i) {};
    }{}
    \ifthenelse{\vrs = 1}
    {
      \node[smallnode,draw] at (rcurvenode_\i) {};
    }{}
  }

  \node[smallnode,draw,fill=gray] at (lcurvenode_\specpath) {};
  \node[smallnode,draw,fill=gray] at (rcurvenode_\specpath) {};
  
  \node[smallnode,draw,label=above:$a$] at (lroof) {};
  \node[smallnode,draw,label=above:$b$] at (rroof) {};

  \foreach \i in {1,...,\npaths}
  {
    \foreach \j in {1,...,\pathlength}
    {
      \node[bignode,draw] at (bignode_\i_\j) {};
    }
  }
  
  \draw[decorate,decoration=brace] (\pathlength,-\npaths-0.5) -- (1,-\npaths-0.5);
  \node at (\pathlength/2+0.5,-\npaths-0.8) {$2\ell$};
  \draw[decorate,decoration=brace] (-\npaths*\curveradiusA+0.7cm,-\npaths) -- +(0,\npaths-1);
  \node at (-\npaths*\curveradiusA,-\npaths/2-0.5) {$h+1$};
\end{tikzpicture}


  \caption{Lower bound construction: Nodes depicted by large circles
    have weight $w$ (heavy nodes), nodes depicted by small circles
    have weight $1$ (light nodes). The graph consists of $h+1$ paths,
    each consisting of $2\ell$ heavy nodes ($h=\ell=6$ in the
    example). Assuming that paths are numbered from $0$ to $h$ from
    top to down. Then, the left-most node on path $0$ is directly
    connected to the left-most node of path $x$ for every $x\not\in
    X$. For $x\in X$, the left-most node of path $0$ is connected to
    the left-most node of path $x$ through an intermediate node of
    weight $1$. The right-most nodes are connected in the same way by
    using the set $Y$. In the figure, we have $X=\set{2,3,5,6}$ and
    $Y=\set{1,4,5}$. The node corresponding to element $5$ in the
    intersection is marked in grey. In addition, nodes $a$ and $b$ are
    used to obtain a network with small diameter.}
  \label{fig:lowerboundgraph}
\end{figure}

We first state an important structural property of graph $H(X,Y)$. In
the following, the size of a vertex cut $S$ of the weighted graph
$H(X,Y)$ is the total weight of the nodes in $S$.

\begin{lemma}\label{lemma:weightedlowerstruct}
  Consider the graph $H(X,Y)$ and assume that $|X\cap Y|\leq 1$. Then,
  if $X$ and $Y$ are disjoint, every vertex cut of graph $H(X,Y)$
  contains a node of weight $w$ (and thus has size at least $w$) and
  if $X\cap Y=\set{z}$ for some $z\in [h]$, the smallest vertex cut of
  $H(X,Y)$ has size $4$ and it consists of the nodes $a$, $b$, $u_z$,
  and $v_z$. In addition, in the second case, every vertex cut of
  $H(X,Y)$ that does not contain $a$, $b$, $u_z$, and $v_z$ contains a
  node of weight $w$. Further, the diameter of $H(X,Y)$ is at most
  $3$.
\end{lemma}
\begin{proof}
  Let us first consider the case $X\cap Y=\emptyset$. In that case,
  for every $z\in[h]$, we either have $z\not\in X$ or $z\not\in Y$. If
  $z\not\in X$, node $(0,1)$ is directly connected to node $(z,1)$, if
  $z\not\in Y$, node $(0,2\ell)$ is directly connected to node
  $(z,2\ell)$. In both cases the path consisting of the nodes $(z,p)$
  for $p\in[2\ell]$ is directly connected to the top path. As this is
  the case for every $z\in[h]$, all $h+1$ paths are directly connected
  to each other and therefore all the nodes of weight $w$ induce a
  connected subgraph. As all other nodes are connected to some node of
  weight $w$, every vertex cut has to contain at least one node of
  weight $w$ and thus, the claim for the case where $X$ and $Y$ are
  disjoint follows.

  For the case, where $X$ and $Y$ intersect in a single element $z$,
  let us consider the path consisting of the nodes $(z,p)$ for
  $p\in[2\ell]$. All the nodes of the path are either connected to
  node $a$ or to node $b$. In addition to this, only the first node
  $(z,1)$ and the last node $(z,2\ell)$ of the path are connected to
  additional nodes. As $z\in X$ and $z\in Y$, node $(z,1)$ is
  connected to $(0,1)$ through node $u_z$ and node $(z,2\ell)$ is
  connected to node $(0,2\ell)$ through node $v_z$. Consequently, by
  removing nodes $a$, $b$, $u_z$, and $v_z$, path $z$ (consisting of
  the nodes $(z,p)$) is disconnected from the rest of the graph. The
  four nodes therefore form a vertex cut of size $4$.

  Now, let us consider any other vertex cut $S\subseteq V_H(X,Y)$ that
  does not contain all of these four nodes. We want to show that $S$
  needs to contain at least one node of weight $w$. For contradiction,
  assume that $S$ contains only nodes of weight $1$. Because for every
  $z'\in[h]\setminus\set{z}$, $z'\not\in X$ or $z'\not\in Y$, the
  same argument as in the $X\cap Y=\emptyset$ case shows that every
  path $z'\in[h]\setminus\set{z}$ is directly connected to path
  $0$. As by assumption also one of the nodes $a$, $b$, $u_z$, or
  $v_z$ is not in $S$, also path $z$ is still connected to the other
  paths. Again since all weight $1$ nodes are directly connected to a
  weight $w$ node, this implies that the nodes $V_H(X,Y)\setminus S$
  induce a connected subgraph, a contradiction to the assumption that
  $S$ contains only nodes of weight $1$.

  It remains to show that the diameter of $H(X,Y)$ is $3$. This
  follows because every node is either directly connected to node $a$
  or to node $b$ and there also is an edge between nodes $a$ and $b$.
\end{proof}

We conclude the discussion on the lower bound construction by finally
also introducing a family $\mathcal{G}$ of unweighted graphs. Given
the three integer parameters $h$, $\ell$, and $w$, there is a
one-to-one correspondence between the graphs of
$\mathcal{H}$ and $\mathcal{G}$. Also in $\mathcal{G}$, there is a
graph $G(X,Y)$ for every possible set disjointness input
$(X,Y)\in[h]^2$. Given $H(X,Y)$, $G(X,Y)$ is obtained by using the
following transformation:
\begin{enumerate}
\item Each node of weight $w$ in $H(X,Y)$ is replaced by a clique of
  size $w$.
\item Each edge of $H(X,Y)$ is replaced by a complete bipartite
  subgraph.\footnote{Hence, edges between two nodes of weight $w$ are
    replaced by a subgraph isomorphic to $K_{w,w}$ and edges between a
    node of weight $1$ and a node of weight $w$ are replaced by a
    subgraph isomorphic to $K_{1,w}$.}
\end{enumerate}
Note that while graphs in $\mathcal{H}$ have $\Theta(h\ell)$ nodes,
graphs in $\mathcal{G}$ have $\Theta(h\ell w)$ nodes. The statements of
Lemma \ref{lemma:weightedlowerstruct} hold in exactly the same way for
graphs of $\mathcal{G}$.

\begin{lemma}\label{lemma:lowerstruct}
  Consider the graph $G(X,Y)$ and assume that $|X\cap Y|\leq 1$. Then,
  if $X$ and $Y$ are disjoint, every vertex cut of graph $G(X,Y)$ has
  size at least $w$ and if $X\cap Y=\set{z}$ for some $z\in [h]$, the smallest vertex cut of
  $G(X,Y)$ has size $4$ and it consists of the nodes $a$, $b$, $u_z$,
  and $v_z$. In addition, in the second case, every vertex cut of
  $G(X,Y)$ that does not contain $a$, $b$, $u_z$, and $v_z$ has size
  at least $w$. Further, the diameter of $G(X,Y)$ is at most $3$.
\end{lemma}
\begin{proof}
  Let $V(X,Y)$ be the set of nodes of $G(X,Y)$ and consider a vertex
  cut $S\subseteq V(X,Y)$ of $G(X,Y)$. Hence, removing the nodes of
  $S$ disconnects the remainder of $G(X,Y)$ into at least $2$
  components. 
  Let $A\subseteq V(X,Y)$ be the $w$ nodes of a clique of size $w$
  corresponding to one of the weight $w$ nodes in $H(X,Y)$ and assume
  that $|S\cap A|\in\set{1,\dots,w-1}$ (i.e., $S$ contains some, but
  not all the nodes of $A$). We first observe that if $S$ is a vertex
  cut, the set $S\setminus A$ is also a vertex cut. Because all edges
  of $H(X,Y)$ are replaced by complete bipartite subgraphs in
  $G(X,Y)$, a single node of $A$ connects the same nodes to each other
  as all the nodes of $A$ do. Given a vertex cut $S$ of $G(X,Y)$, we
  can therefore always find a vertex cut $S'\subseteq S$ of $G(X,Y)$
  such that $S'$ contains either none or all the nodes of each of the
  cliques of size $w$ corresponding to the weight $w$ nodes of
  $H(X,Y)$. Let us call such a vertex cut $S'$, a reduced vertex
  cut. Note that there is a one-to-one correspondence between the
  vertex cuts of $H(X,Y)$ and the reduced vertex cuts of $G(X,Y)$.
  
  The first part of Lemma \ref{lemma:weightedlowerstruct} therefore
  implies that if $X\cap Y=\emptyset$, every reduced vertex cut of
  $G(X,Y)$ contains at least one complete clique of size $w$ and it
  therefore has size at least $w$. Hence, using the above observation,
  we also get that every vertex cut of $G(X,Y)$ has size at least $w$.

  If $X$ and $Y$ intersect in a single element $z\in [h]$, Lemma
  \ref{lemma:weightedlowerstruct} implies that nodes $a$, $b$, $u_z$,
  and $v_z$ form a (reduced) vertex cut of size $4$ (note that the
  four nodes all have weight $1$ in $H(X,Y)$). Also, if a reduced
  vertex cut $S$ of $G(X,Y)$ does not contain all the four nodes, Lemma
  \ref{lemma:weightedlowerstruct} implies that contains at least one
  complete clique of size $w$ and thus every vertex cut that does not
  contain all the four nodes has size at least $w$.

  Finally, we get that graph $G(X,Y)$ has diameter $3$ by using
  exactly the same argument as for $H(X,Y)$.
\end{proof}

\subsection{Reduction}
\label{sec:lowerreduction}

We next show how an efficient distributed algorithm to approximate the
vertex connectivity or find a small vertex cut in networks of the
family $\mathcal{G}$ can be used to get a two-party set disjointness
protocol with low communication complexity. We first show that for
$T<\ell$, any $T$-round distributed protocol on a graph
$G(X,Y)\in\mathcal{G}$ can be simulated in a low communication
public-coin two-party protocol by Alice and Bob, assuming that Alice
knows the inputs of all except the right-most nodes of $G(X,Y)$ and
Bob knows the inputs of all except the left-most nodes of
$G(X,Y)$. Because only these nodes are used to encode the set
disjointness instance $(X,Y)$ into $G(X,Y)$, together with Lemma
\ref{lemma:lowerstruct}, this allows to derive a lower bound on the
time to approximate the vertex connectivity or finding small vertex
cuts. For convenience, we again first state the simulation result for
graphs $H(X,Y)\in\mathcal{H}$. The proof of the following lemma is
done in a similar way as the corresponding simulation in
\cite{dassarma12}. For all $r\in\set{0,\dots,\ell-1}$, we define set
$V_A(r)$ and $V_B(r)$ as follows.
\begin{eqnarray*}
  V_A(r) & := & \set{a}\cup V_X \cup
  \set{(p,q)\in\set{0,\dots,h}\times[2\ell]: q < 2\ell-r},\\
  V_B(r) & := & \set{b}\cup V_Y \cup
  \set{(p,q)\in\set{0,\dots,h}\times[2\ell]: q > r+1}.
\end{eqnarray*}

\begin{lemma}\label{lemma:weightedsimulation}
  Let $T\leq\ell$ be an integer and let $\mathcal{A}$ be a $T$-round
  randomized distributed algorithm on graphs
  $H(X,Y)\in\mathcal{H}$. Assume that in each round, nodes $a$ and $b$
  locally broadcast a message of at most $B$ bits to their neighbors
  (other nodes are not restricted). Further, assume that Alice knows
  the initial states of nodes $V_A(0)$ and Bob knows the initial
  states of nodes $V_B(0)$. Then, Alice and Bob can simulate
  $\mathcal{A}$ using a randomized public-coin protocol such that:
  \begin{enumerate}
  \item At the end, Alice knows the states of nodes $V_A(T)$ and Bob
    knows the states of nodes $V_B(T)$
  \item Alice and Bob need to exchange at most $2B\cdot T$ bits.
  \end{enumerate}
\end{lemma}
\begin{proof}
  First note that we can use the public randomness to model the
  randomness used by all the nodes of $H(X,Y)$. Hence, the random bits
  used by the nodes in the distributed protocol $\mathcal{A}$ is
  publicly known. We next describe a two-party protocol in which Alice
  and Bob simulate $\mathcal{A}$ in a round-by-round manner such that
  for all rounds $0\leq r<\ell$, after simulating round $r$ (or
  initially for $r=0$),
  \vspace*{2mm}
  \begin{itemize}
  \item[(I)] Alice knows the states of nodes in $V_A(r)$.
  \item[(II)] Bob knows the states of nodes in $V_B(r)$.
  \item[(III)] Alice and Bob have exchanged at most $2B\cdot r$ bits.
  \end{itemize}
  \vspace*{2mm}
  We prove (I), (II), and (III) by induction on $r$.

  \paragraph{Induction Base} For $r=0$, statements (I)--(III) follow
  directly from the assumptions about the initial knowledge of Alice
  and Bob.

  \paragraph{Induction Step} For $r\geq 1$, assume that (I)--(III)
  hold for $r< r'$, where $r'\in \set{0,\dots,T-1}$ so that we need to
  show that it also holds for $r=r'$. We need to show how Alice an Bob
  can simulate round $r$. In order for (III) to hold, Alice an Bob can
  exchange at most $2B$ bits for the simulation of round $r$. In order
  to satisfy (I), observe the following. We need to show that after
  the simulation of round $r$, Alice knows the states of all nodes in
  $V_A(r)$. By the induction hypothesis, we know that Alice knows the
  states of the nodes $V_A(r-1)\supset V_A(r)$ after round
  $r-1$. Hence, in addition, in order to be able to compute the states
  of the nodes $V_A(r)$ after round $r$, Alice needs to know all the
  messages that nodes in $V_A(r)$ receive in round $r$. She therefore
  needs to know all the messages that are sent by neighbors of nodes
  in $V_A(r)$ in round $r$. The set of neighbors of nodes in $V_A(r)$
  consists of the nodes $V_A(r-1)$ and of node $b$. Note that in
  particular, because $T\leq \ell$, $V_A(r-1)$ also contains all the
  neighbors of node $a\in V_A(r)$.  Except for node $b$, Alice thus
  knows the state of all neighbors of node in $V_A(r)$ at the
  beginning of round $r$ and she therefore also knows the messages
  sent by these nodes in round $r$. In order complete her simulation
  of round $r$, she therefore only needs to learn the message (of at
  most $B$ bits) sent by node $b$ in round $r$. By the induction
  hypothesis, Bob knows the content of this message and can send it to
  Alice. Similarly, Bob can also compute the states of all nodes in
  $V_B(r)$ at the end of round $r$ if Alice sends the round $r$
  message of node $a$ to Bob. This completes the proof of the
  induction step and thus also the proof of the lemma.
\end{proof}

An analogous lemma can also be shown for graphs
$G(X,Y)\in\mathcal{G}$. Here, we define $V'_A(r)$ and $V'_B(r)$ to be
the node sets corresponding to $V_A(r)$ and $V_B(r)$. That is,
$V'_A(r)$ contains all weight $1$ nodes of $V_A(r)$ and all the $w$
nodes of each clique of size $w$ corresponding to a weight $w$ node
in $V_A(r)$. The set $V'_B(r)$ is defined analogously. Based on the
argument for $\mathcal{H}$, we then directly obtain the following
statement for graphs in $\mathcal{G}$.

\begin{lemma}\label{lemma:simulation}
  Let $T\leq\ell$ be an integer and let $\mathcal{A}$ be a $T$-round
  randomized distributed algorithm on graphs
  $G(X,Y)\in\mathcal{G}$. Assume that in each round, nodes $a$ and $b$
  locally broadcast a message of at most $B$ bits to their neighbors
  (other nodes are not restricted). Further, assume that Alice knows
  the initial states of nodes $V'_A(0)$ and Bob knows the initial
  states of nodes $V'_B(0)$. Then, Alice and Bob can simulate
  $\mathcal{A}$ using a randomized public-coin protocol such that:
  \begin{enumerate}
  \item At the end, Alice knows the states of nodes $V'_A(T)$ and Bob
    knows the states of nodes $V'_B(T)$
  \item Alice and Bob need to exchange at most $2B\cdot T$ bits.
  \end{enumerate}
\end{lemma}
\begin{proof}
  The proof is done in the same way as for Lemma \ref{lemma:weightedsimulation}.
\end{proof}

We are now ready to prove the lower bound Theorem \ref{thm:vertexlowerbound}.

\begin{proof}[{\bf Proof of Theorem \ref{thm:vertexlowerbound}}]
  Let us first assume that there is a randomized $T$-round \LBM\ model
  protocol $\mathcal{A}$ that allows distinguish graphs of vertex
  connectivity at most $k$ from graphs of vertex connectivity at least
  $k\alpha$. Alice and Bob can use protocol $\mathcal{A}$ to solve
  the set disjointness problem as follows. Assume that Alice and Bob
  are given inputs $X\subseteq [h]$ and $Y\subseteq [h]$ for some
  positive integer $h$ with the promise that $X$ and $Y$ intersect in
  at most $1$ value. We pick $\ell=h/\log n$ and $w=\alpha k + 1$ and
  we consider the graph $G(X,Y)$ with parameters $h$, $\ell$, and
  $w$. Assume that $T< \ell$. Note that except for the very first
  cliques of each of the paths of $G(X,Y)$ and the very last cliques
  of each of the paths of $G(X,Y)$, the graph $G(X,Y)$ does not depend
  on $X$ and $Y$. Hence, Alice knows the initial states of all nodes
  in $V'_A(0)$ and Bob knows the initial states of all nodes in
  $V'_B(0)$. Using Lemma \ref{lemma:simulation}, Alice and Bob can
  therefore simulate the $T$ rounds of $\mathcal{A}$ by exchanging at
  most $2BT$ bits such that in the end for all nodes $v$ of $G(X,Y)$,
  either Alice or Bob knows the final state of $v$. Alice and Bob
  therefore definitely learn the approximation of the vertex
  connectivity computed by $\mathcal{A}$. By Lemma
  \ref{lemma:lowerstruct}, if $X\cap Y=\emptyset$, the vertex
  connectivity of $G(X,Y)$ is at least $w\geq \alpha k + 1$ and if
  $X\cap Y\neq\emptyset$, the vertex connectivity of $G(X,Y)=4\leq
  k$. An $\alpha$-approximation of the vertex connectivity therefore
  allows Alice and Bob to solve the set disjointness instance
  $(X,Y)$. As by the set disjointness lower bound of
  \cite{razborov92}, solving set disjointness of sets from the
  universe $[h]$ requires Alice and Bob to exchange at least
  $\Omega(h)$ bits, we get that $2TB=\Omega(h)$ and thus
  $T=\Omega(h/B)=\Omega(h/\log n)$. Together with
  $n=\Theta(h\ell\alpha k)$, the claimed lower bound follows.

  We directly also get a lower bound on computing a fractional dominating tree
  packing (or a fractional connected dominating set packing) of size
  at least $k/\alpha$ because the size of such a packing leads to the
  corresponding approximation of the vertex connectivity.

  To prove the lower bound on finding small vertex cuts, we consider
  instances $(X,Y)$ for which $|X\cap Y|=1$. Let the element in the
  intersection $X\cap Y$ be $z$. Note that by Lemma
  \ref{lemma:lowerstruct}, in that case the vertex connectivity of
  $G(X,Y)$ is $4$ and every vertex cut of size at most $\alpha k<w$
  needs to contain the nodes $u_z$, $v_z$, $a$, and $b$. Hence, an
  algorithm that outputs a vertex cut of size $s\leq
  \min\set{\delta\sqrt{n/(\alpha k \log n)},\alpha k}$ has to output a
  node set $S$ of size $s$ such that in particular $u_z,v_z\in
  S$. Since $S$ contains at most $s-2$ other nodes $u_x\in V_X$ or
  $v_y\in V_Y$, the same reduction as above allows Alice and Bob to
  output a set of at most $s-1$ elements from $[h]$ such that $z$ is
  contained in this set. For a sufficiently small constant $\delta>0$,
  we have seen that for this, Alice and Bob also need to exchange at
  least $\Omega(h)$ bits.
\end{proof}


\end{document}